\documentclass[12pt]{article}

\RequirePackage[OT1]{fontenc}
\RequirePackage{amsthm,amsmath,amssymb,amsfonts,graphicx,subfigure,caption,bm,fullpage,wrapfig,setspace}
\RequirePackage[colorlinks,citecolor=blue,urlcolor=blue]{hyperref}

\numberwithin{equation}{section}
\theoremstyle{plain}
                          \newtheorem{thm}{Theorem}
                          \newtheorem{lem}{Lemma}
                          
\theoremstyle{remark}         
\theoremstyle{definition} \newtheorem{defn}{Definition}
                          \newtheorem{exam}{Example}

\newcommand\om{\Omega}
\newcommand\real{\mathbb{R}}
\newcommand\A{\mathcal{A}}
\newcommand\B{\mathcal{B}}
\newcommand\C{\mathcal{C}}

\newcommand\bx{\bm{x}}

\newcommand\bF{\bm{F}}

\newcommand\bnu{\bm{\nu}}

\newcommand\bthe{\bm\theta}

\newcommand\bgam{\bm\gamma}

\newcommand\bphi{\bm\phi}

\newcommand\tgam{\tilde{\gamma}}

\newcommand\bi{\begin{itemize}}
\newcommand\ei{\end{itemize}}

\def\I{{\mathbf{1}}}

\def\woMR#1{\w@MR#1MR#1MR\relax}%
\def\w@MR#1MR#2MR#3\relax{#2}

\def\@MR#1 #2\relax#3{%
 \href{http://www.ams.org/mathscinet-getitem?mr=#1}%
 {\MRfixed{#3}}}%

\def\MRfixed{MR\woMR}%
\let\MR\MRfixed

\usepackage{natbib}

\usepackage{varioref}		
\labelformat{section}{Section~#1}
\labelformat{figure}{Figure~#1}
\labelformat{table}{Table~#1}	

\title{Markov adaptive P\'olya trees and multi-resolution adaptive shrinkage in nonparametric modeling}
\author{Li Ma}

\begin{document}
\maketitle

\doublespacing

\begin{abstract}
We introduce a hierarchical nonparametric model for probability measures based on a multi-resolution transformation of probability distributions. The model allows a varying amount of shrinkage to be applied to data features of different scales and/or at different locations in the sample space, and the varying shrinkage level is locally adaptive to the empirical behavior of the data. Moreover, the model's hierarchical design---through a latent Markov tree structure---allows borrowing of information across locations and scales in setting the adaptive shrinkage level. Inference under the model proceeds efficiently using general recipes for conjugate hierarchical models. We illustrate the work of the model in density estimation and evaluate its performance through simulation under several schematic scenarios carefully designed to be representative of a variety of applications. We compare its performance to those of several state-of-the-art nonparametric models---the P\'olya tree, the optional P\'olya tree, and the Dirichlet process mixture of normals. In addition, we establish several important theoretical properties for the model including absolute continuity, full nonparametricity, and posterior consistency. 
\end{abstract}
\newpage

\section{Introduction}
\label{sec:intro}
\vspace{-0.8em}

In his seminal works that jump-started modern Bayesian nonparametric inference, \cite{ferguson:1973,ferguson:1974} formalized the notion of a Dirichlet process (DP) and introduced a tail-free process that contains the DP as a special case. This tail-free process was later named the P\'olya tree (PT) due to its relationship to the P\'olya urn model \citep{mauldin:1992}, and was popularized in the 1990's by a sequence of works \citep{lavine:1992,mauldin:1992,lavine:1994} that investigated its various theoretical properties.

The PT produces probability measures through a multi-resolution generative procedure. 
The sample space is recursively bisected into smaller and smaller sets, and for each set $A$ that arises during the partition, the probability assigned to $A$ is randomly split between its two children $A_l$ and $A_r$ through the drawing of a Beta random variable corresponding to the proportion of mass assigned to $A_l$.

There is an impressive resemblance between the inference schema of the PT and that of other multi-resolution inference methods such as wavelet denoising \citep{draper:1999}. In particular, inference with the PT and wavelet analysis both adopt a ``divide-and-conquer'' strategy: a nonparametric quantity of interest is transformed into a collection of coefficients defined on a multi-resolution tree. These coefficients characterize the shape of the nonparametric quantity at different locations and scales. More specifically, for wavelets, each coefficient---called a wavelet coefficient (WC)---specifies the local contrast of the function value, while for the PT, each coefficient---which we shall refer to (and define formally later) as the probability assignment coefficient (PAC)---characterizes how probability mass is assigned locally. Estimating the underlying nonparametric quantity then proceeds through inferring the WCs and PACs. \ref{tab:mr} summarizes the analogy between multi-resolution density estimation using PT and wavelet-based function estimation.

\begin{table}[t]
    \begin{tabular}{ c | cc }
     & Wavelet denoising & PT-type density estimation \\
    \hline
    Quantity of interest & mean function & probability density\\
    Data & random function(s) with noise & i.i.d.\ observations from the density\\
    Unit of inference  & wavelet coefficient & probability assignment coefficient\\
    Unit model & Gaussian experiment & binomial experiment\\
  \hline\hline
    \end{tabular}

\vspace{0.8em}

\caption{Analogy between wavelet denoising and PT-type density estimation}
\label{tab:mr}
\vspace{-0.3em}
\end{table}

It is well-known in wavelet analysis that a key to effective estimation is appropriate shrinkage on the WCs  to differentiate ``signals'' from ``noise'' \citep{donoho:1994}. Numerous shrinkage methods, both frequentist and Bayesian, have been proposed for achieving data-adaptive shrinkage for wavelet denoising \citep{vidakovic:1999}. We note that the same is true for multi-resolution density estimation---appropriate shrinkage on the PACs is also critical for effective inference. However, existing methods such as the standard PT allows no adaptivity in shrinkage. In particular, the PT model places independent Beta priors on the PACs, which applies a prespecified, fixed amount of shrinkage determined by the prior Beta variance to the PACs. This is analogous to placing independent Gaussian priors with fixed variances on the WCs in a wavelet analysis. The following example illustrates how the lack of adaptivity in shrinkage can result in poor inference with the PT, especially when the underlying distribution contains structures of different scales.
\vspace{-0.5em}

\begin{exam}
\label{ex:motive}
We simulate 750 i.i.d data from the following mixture distribution on $[0,1]$
\vspace{-1.7em}

\[
0.1\,{\rm U}(0,1) + 0.3\,{\rm U}(0.25,0.5) + 0.4\,{\rm Beta}_{(0.25,0.5)}(2,2) + 0.2\,{\rm Beta}(6000,4000)
\]
\vspace{-2.7em}

\noindent where ${\rm Beta}_{(0.25,0.5)}(2,2)$ represents a Beta$(2,2)$ translated and scaled to be supported on the interval (0.25,0.5)---that is, the distribution with density $8(4x-1)(1-2x)$ on $(0.25,0.5)$. \ref{fig:ex1_motivating} illustrates the pdf (red dashed).
The ``hump'' on the interval $(0.25,0.5)$ constitutes a distributional structure of a relatively large scale or low resolution, while the spike given by Beta$(6000,4000)$ constitutes a small-scale or high-resolution feature.

\begin{figure}[ht]
  \centering
  \includegraphics[width=39em]{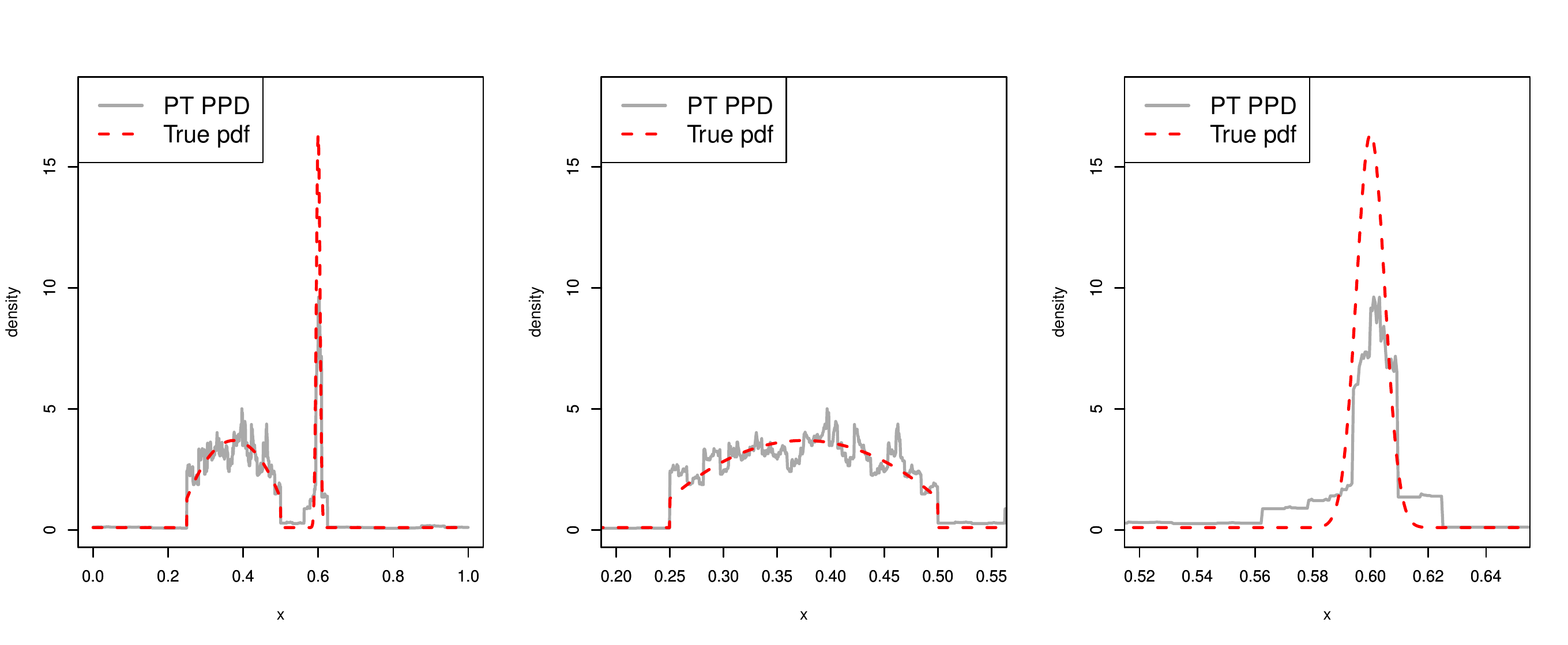}
  \vspace{-1.7em}

  \caption{The true density (red dashed) and the PPD of the PT (gray solid) for Example~\ref{ex:motive}. The middle and right plots give the zoom-in views of the low- and high-resolution features.}
  \label{fig:ex1_motivating}
\end{figure}

Let us place a PT prior on the underlying distribution corresponding to a Beta($k^2,k^2$) prior on for PACs at level $k$, which is the most common specification in applications of the PT. (See \cite{lavine:1992,walker:1999a,hanson:2002,hanson:2006,holmes:2009} for example.) The gray solid curves in \ref{fig:ex1_motivating} shows the posterior predictive density (PPD) of the PT. The middle and right plots give the zoom-in views of the large-scale and small-scale features respectively. We see that overall the PT induces a decent amount of shrinkage for capturing the low-resolution feature (middle), as the general shape of that feature is adequately recovered in the PPD. For the high-resolution feature, however, it results in too much shrinkage and thus under-estimation of the mode (right). Interestingly, if we zoom into higher resolutions within the low-resolution feature (middle), 
the PPD shows jumpy patterns of overfitting, indicating that more shrinkage is needed there at higher resolutions to ensure proper smoothness.
\end{exam}

The above example represents a typical situation---the appropriate amount of shrinkage varies across locations and scales. Only an oracle can {\em a priori} choose the prior variance of the PACs that give the optimal amount of shrinkage at all resolutions and locations. Popular convenient choices such as $k^2$ for the Beta parameters---or any other fixed function of $k$ for that matter---cannot provide the right amount of shrinkage at all locations and scales, just as independent Gaussian priors with pre-specified variances for the WCs are inadequate in wavelet denoising \citep{chipman:1997,vidakovic:1998,clyde:2002,clyde:2000}. 

Our main goal is to incorporate adaptive shrinkage into the multi-resolution nonparametric modeling framework adopted in the PT and related models, which we shall refer to generally as ``PT-type'' models. Both classical thresholding and Bayesian shrinkage methods can be adopted, and we shall take a hierarchical Bayesian approach focusing on constructing hyperpriors in the form of generative models on the shrinkage parameters, i.e.\ the prior variance of the PACs.
Our modeling approach provides a natural way to introducing dependency into the shrinkage levels for different PACs, thereby achieving borrowing of information across locations and scales in determining the level of shrinkage for each PAC. 
In particular, we construct a joint generative model using a latent variable representation that takes the form of a Markov tree \citep{crouse:1998} to achieve {\em stochastically} increasing shrinkage (defined later), which is critical for achieving effective locally adaptive smoothing in density estimation. This is in contrast to the {\em deterministically} increasing shrinkage imposed by the PT with prespecified increasing Beta parameters.

The rest of the work is organized as follows. In \ref{sec:method}, we present the core of our methodology. We begin by viewing the PT from a hierarchical modeling viewpoint, and interpret inference under the PT from a shrinkage perspective. Then we construct a simple hyperprior on the variance of PACs that treats the PACs independently. Finally we extend the hyperprior through a latent variable representation to a Markov tree hyperprior that incorporates dependency in the shrinkage levels. We provide guidelines for prior specification and establish both methodological and theoretical properties of the model. Specifically, we show that the model is fully nonparametric (i.e.\ with full prior support) and enjoys posterior consistency. Moreover, posterior inference under this model can be carried out conveniently using general recipes for conjugate hierarchical models. In particular, the full posterior can be analytically derived and sampled from exactly without resorting to MCMC.  In \ref{sec:numerical_examples} we illustrate how our method works in density estimation and evaluate its performance under various scenarios where the underlying density possesses a variety of features. We also compare our method to two state-of-the-art PT-type multi-resolution models---the PT and the optional P\'olya tree (OPT) \citep{wongandma:2010}---as well as to the very popular Dirichlet process mixture (DPM) of normals \citep{escobar:1995}. We conclude in \ref{sec:discussion} with brief remarks on applications to hypothesis testing and computational efficiency. 

We close this introduction by noting that adaptive shrinkage has been extensively studied in the context of wavelet denoising from both frequentist and Bayesian perspectives. The literature is too enormous to be enumerated. A far-from-exhaustive list of notable examples from the frequentist perspective include \cite{donoho:1994,donoho:1995a,donoho:1995b,donoho:1995c,nason:1995,abramovich:1995,nason:1996,johnstone:1997,donoho:1998,cai:1998,hall:1998,kolaczyk:1999,johnstone:1999,cai:1999,antoniadis:2001,johnstone:2005}, and from the Bayesian perspective include  \cite{chipman:1997,abramovich:1998,clyde:1998,crouse:1998,vidakovic:1998,vannucci:1999,moulin:1999,chang:2000,clyde:2000,romberg:2001,brown:2001,clyde:2002,portilla:2003,morris:2006}. For the particular application of density estimation, there is also a body of literature on wavelet-based methods. See for example \cite{vannucci:1998,donoho:1996,koo:2000,herrick:2001}. 
\vspace{-1.5em}

\section{Method}
\label{sec:method}
\vspace{-1em}

\subsection{Multi-resolution representation of probability distributions}
\vspace{-0.5em}

We start by introducing some basic concepts and definitions that form the building blocks for multi-resolution modeling of probability distributions. Throughout this work, we let $\om$ denote the sample space, which can be finite or a (possibly unbounded) Euclidean rectangle such as an interval in $\real$ or a rectangle in $\real^{p}$. Let $\mu$ be the natural measure associated with $\om$, which is the counting measure if $\om$ is finite and the Lebesgue measure if $\om$ is Euclidean.

Let $\A^1,\A^{2},\ldots,\A^{k},\ldots$ be a {\em sequence of nested dyadic partitions} of $\om$. That is, each $\A^{k}=\{A_{k,1},A_{k,2},\ldots,A_{k,2^{k}}\}$ and it satisfies (i) $\om=\cup_{m=1}^{2^k} A_{k,m}$, (ii) $A_{k,m_1}\cap A_{k,m_2}=\emptyset$ for all $m_1\neq m_2$, and (iii) $A_{k,m}=A_{k+1,2m-1}\cup A_{k+1,2m}$ for all $k=1,2,\ldots$ and $m=1,2,\ldots,2^k$. In other words, the partition $\A_{k+1}$ is obtained by dividing each set in $\A^{k}$ into two children, the {\em left child} $A_{k+1,2m-1}$ and the {\em right child} $A_{k+1,2m}$. We shall call $\A^{k}$ the partition at resolution (or scale) $k$. Also, we let
$\A^{(\infty)}=\cup_{k=1}^{\infty} \A^{k}$,  the totality of all partition sets that arise in all resolution levels. The partition sets form a bifurcating tree, so from now on we shall refer to $\A^{(\infty)}$ as the {\em partition tree}, and each $A$ in $\A^{(\infty)}$ as a {\em node}. Because each $A$ corresponds to a particular location and scale in the tree, we also interchangeably refer to a node as a {\em location-scale combination}.

Given a partition tree $\A^{(\infty)}$ that generates the Borel $\sigma$-algebra, 
one can describe a probability distribution $G$ by specifying how probability mass is split between the left and right children on each node $A$. Let $A_l$ and $A_r$ be the left and right children of a node $A$. We define the {\em probability assignment coefficient} (PAC) for $A$ to be the proportion of probability mass assigned to $A_l$, and denote it as $\theta(A)$. So if the total probability mass on $A$ is $G(A)$ then those assigned to the children are $G(A_l)=G(A)\theta(A)$ and $G(A_r)=G(A)(1-\theta(A))$. 
\begin{lem}
\label{lem:uniqueness_of_PACs}
Given a partition tree $\A^{(\infty)}$ that generates the Borel $\sigma$-algebra, every probability distribution $G$ can be mapped to a collection of PACs $\{\theta(A):A\in\A^{(\infty)}\}$, and the mapping is unique on all $A$s such that $G(A)>0$.
\end{lem}
\vspace{-0.5em}

\noindent Remark: A collection of PACs corresponding to $G$ is given by $\theta(A)=G(A_l)/G(A)$ for all $A$ with $G(A)>0$ and $\theta(A)=0$ otherwise.
\vspace{0.5em}

\ref{fig:mr_decomp}(a) illustrates the transformation of a distribution into PACs. Each PAC specifies the structure of the distribution at a given scale and location.
\begin{figure}[t]
\vspace{-0.5em}
  \begin{center}

    \mbox{
      \subfigure[Transforming a distribution into PACs (red ticks)]{\includegraphics[width=0.57\textwidth,clip=TRUE, trim=55mm 75mm 60mm 45mm]{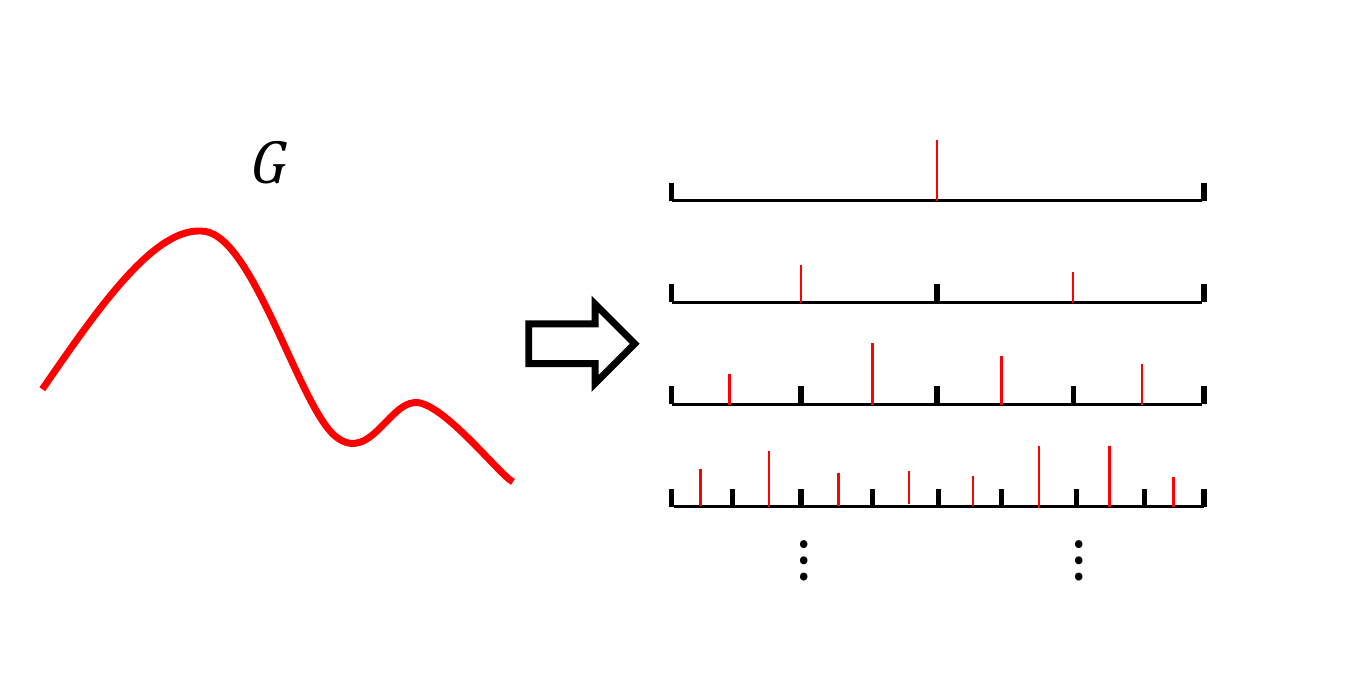}}
      \subfigure[Local binomial experiment]{\includegraphics[width=0.45\textwidth, clip=TRUE,trim=65mm 70mm 70mm 45mm]{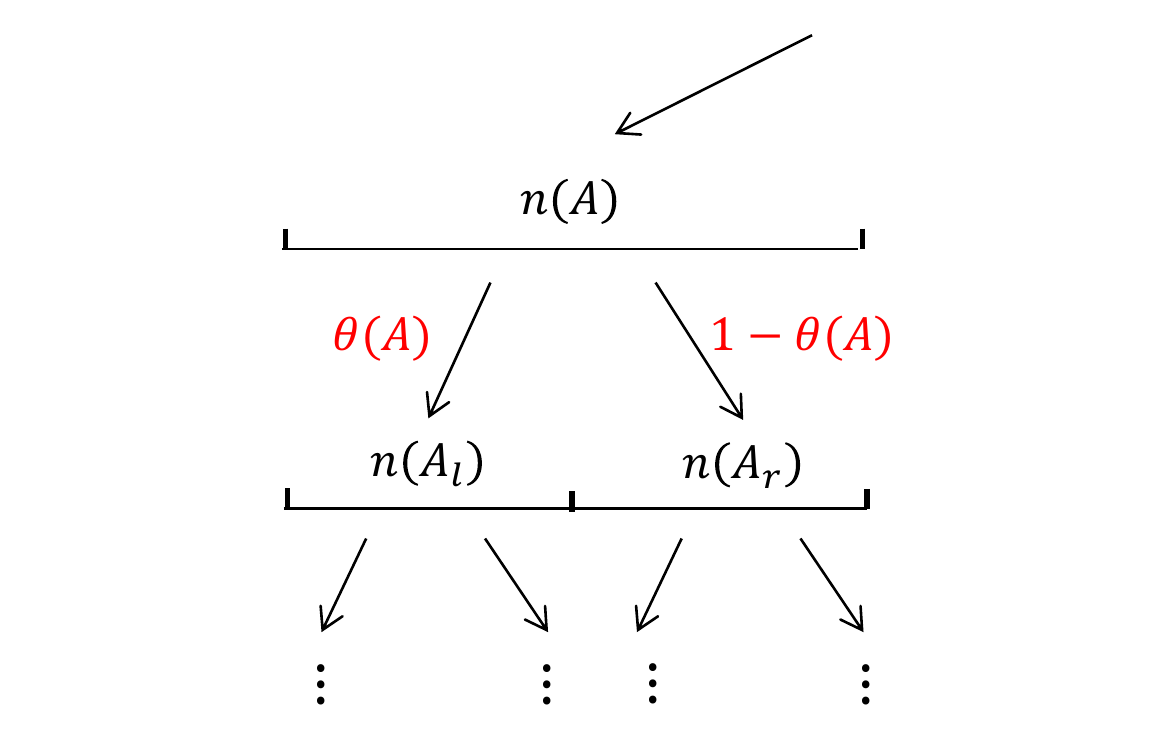}}
   }
    \caption{The divide-and-conquer inference schema.}
    \label{fig:mr_decomp}
  \end{center}
\vspace{-1.5em}
\end{figure}

The lemma implies that inference on a distribution can be achieved by inferring the PACs, which motivates a ``divide-and-conquer'' strategy for nonparametric inference. More specifically, the multi-resolution transformation of $G$ into PACs induces a corresponding decomposition of the statistical experiment that generates an i.i.d.\ sample of size $n$ from $G$. In particular, the experiment is divided into a collection of local binomial experiments carried out sequentially in the order of the resolution: $\left\{ n(A_l) \sim {\rm Binomial}\left(n(A),\theta(A)\right): A\in \A^{k}\right\}$ for $k=1,2,\ldots$ where $n(A)$ is the number of data points in $A$ arising from the binomial experiment on $A$'s parent in the previous resolution $k-1$, except that for $k=1$, $n(A)=n$ by design.  See \ref{fig:mr_decomp}(b) for an illustration of the local binomial experiment.

Accordingly, inferring a distribution through the PACs is divided into inference on the success probabilities of a collection of sequential binomial experiments. Viewed this way, the PT model provides a simple solution to this problem---it places independent conjugate Beta priors on the successes probabilities. The posterior conjugacy of the PT follows immediately from the Beta-binomial conjugacy. Thus inference under PT model is analytically tractable and computationally efficient.

\vspace{-1em}

\subsection{Multi-resolution shrinkage and the adaptive P\'olya tree}
\vspace{-0.5em}

One can now understand the shrinkage property of the PT by viewing each of the binomial experiment from a shrinkage perspective. Under the PT model, $\theta(A)\sim{\rm Beta}(\alpha_l(A),\alpha_r(A))$ for all $A\in\A^{(\infty)}$. (A popular specification has $\alpha_l(A)=\alpha_r(A)=k^2$ for $A\in \A^k$.) We shall prefer an alternative parametrization of Beta distributions in terms of a mean parameter $\theta_0(A)=\alpha_l(A)/(\alpha_l(A)+\alpha_r(A))$ and a precision parameter $\nu(A)=\alpha_l(A)+\alpha_r(A)$. The posterior distribution of $\theta(A)$ is still Beta with mean parameter
\vspace{-1em}

\[
\tilde{\theta}_0(A) = E(\theta(A)|\bx) = \theta_0(A)\cdot \frac{\nu(A)}{\nu(A)+n(A)} + \frac{n(A_l)}{n(A)} \cdot \frac{n(A)}{\nu(A)+n(A)}
\]
and precision parameter $\tilde{\nu}(A)=\nu(A)+n(A)$.
The posterior mean is a weighted average between $\theta_0(A)$, or the prior mean, and $n(A_l)/n(A)$, or the PAC on $A$ of the empirical distribution. The level of shrinkage for $\theta(A)$ is controlled by the precision parameter $\nu(A)$, and thus we shall refer to $\nu(A)$ also as the {\em shrinkage parameter}. 

The prior mean of the PT is the probability distribution corresponding to the collection of PACs $\{\theta_0(A):A\in \A^{(\infty)}\}$, which we call $Q_0$, while its posterior mean is the distribution corresponding to $\{\tilde{\theta}_0(A):A\in\A^{(\infty)}\}$ as the PACs. Intuitively, the posterior mean of a PT is a weighted average between $Q_0$ and the empirical distribution, but the weighting is scale and location dependent. \ref{fig:graphical}(a) provides a graphical model representation of the PT. 

\begin{figure}[t]
  \begin{center}
    \vspace{-2em}

    \mbox{
      \subfigure[P\'olya tree]{\includegraphics[width=0.26\textwidth,clip=TRUE, trim=75mm 28mm 70mm 70mm]{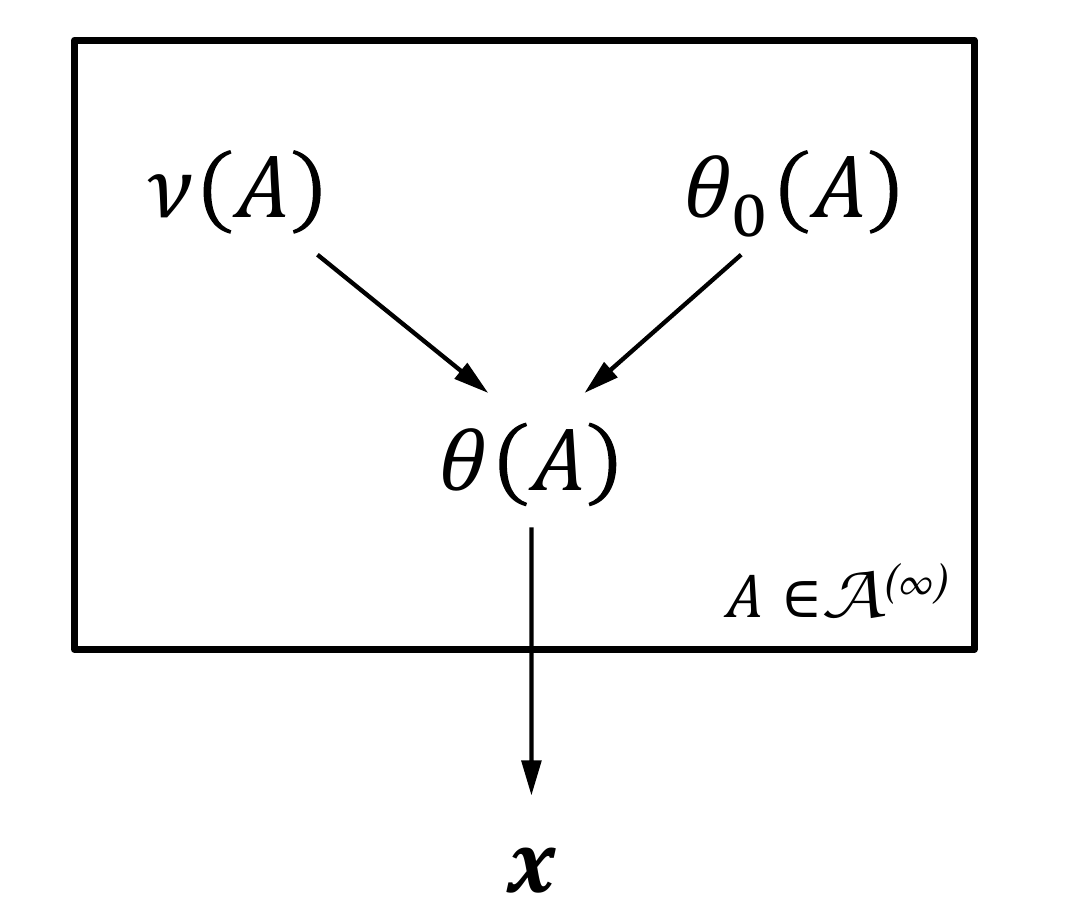}}\hspace{2.8em}

      \subfigure[Adaptive P\'olya tree]{\includegraphics[width=0.26\textwidth,clip=TRUE, trim=75mm 28mm 75mm 50mm]{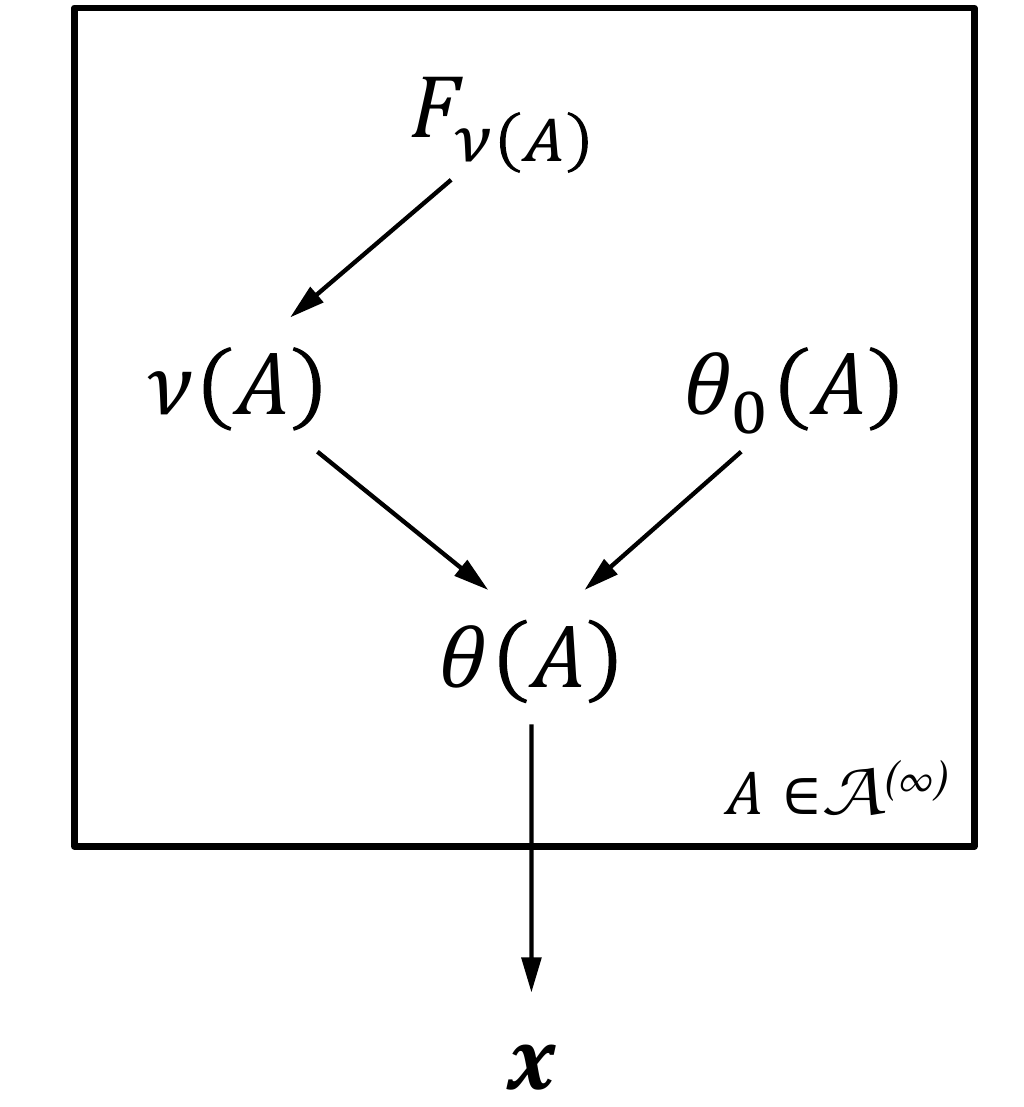}}\hspace{1.8em}

      \subfigure[Markov adaptive P\'olya tree]{\includegraphics[width=0.36\textwidth, clip=TRUE, trim=20mm 30mm 70mm 20mm]{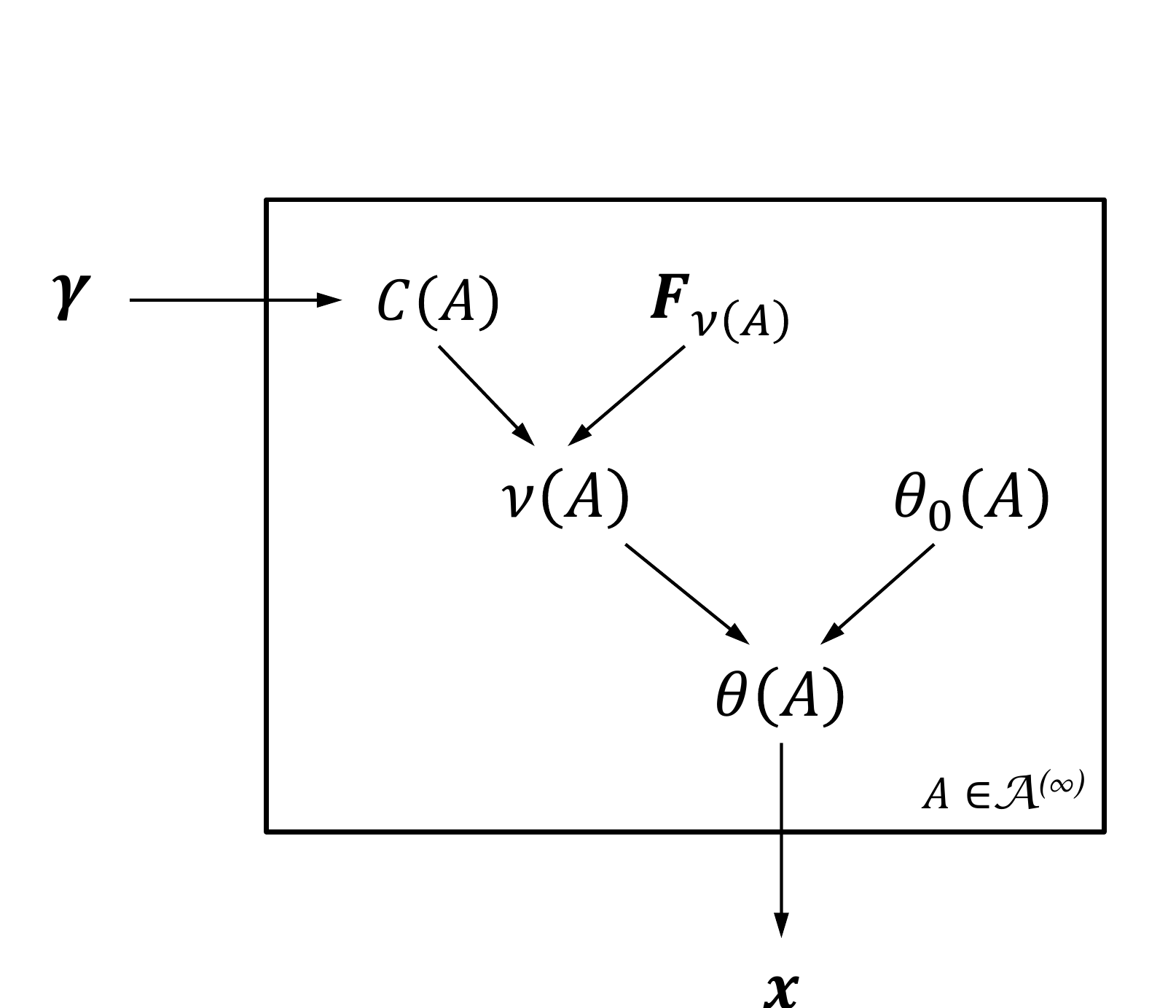}}
   }
    \caption{Graphical representation of PT (left), APT (middle), and Markov-APT (right).}
    \label{fig:graphical}
  \end{center}
\vspace{-1em}
\end{figure}

From a hierarchical Bayesian perspective, we can incorporate adaptivity into the shrinkage by placing a hyperprior $F_{\nu(A)}$ on $\nu(A)$ thereby allowing the appropriate level of shrinkage for $\theta(A)$ to be inferred from the data. This leads to a generative hierarchical model specified by  
$F_{\bnu}=\{F_{\nu(A)}:A\in\A^{(\infty)}\}$ and $Q_0$. 
Using $\bphi=(F_{\bnu},Q_0)$ to represent the totality of all hyperparameters, we can write the model as
\vspace{-3.5em}

\begin{align*}
\nu(A)\,|\,\bphi &\sim F_{\nu(A)}\\ 
\theta(A)\,|\,\bphi,\bm{\nu} & \sim {\rm Beta}(\theta_0(A)\nu(A),(1-\theta_0(A))\nu(A)) 
\end{align*}
\vspace{-3.2em}

\noindent for all $A\in\A^{(\infty)}$, where  
$\bm{\nu}=\{\nu(A):A\in\A^{(\infty)}\}$, the collection of the Beta precision parameters. \ref{fig:graphical}(b) provides a graphical representation of this model.

We consider priors $F_{\nu(A)}$ supported on $(0,\infty]$. 
Note that we allow $\nu(A)=\infty$, in which case ${\rm Beta}(\theta_0(A)\nu(A),(1-\theta_0(A))\nu(A))$ is a point mass at $\theta_0(A)$, corresponding to complete shrinkage of $\theta(A)$ to the prior mean $\theta_0(A)$.
The first natural question is whether this model always generates well-defined probability measures. In other words, given a set of PACs $\{\theta(A):A\in \A^{(\infty)}\}$ arising from the above model, does there (almost surely) exist a distribution $G$ such that $G(A_l|A)=\theta(A)$ for all $A\in \A^{(\infty)}$? 
The answer is positive by Theorem~3.3.2 in \cite{ghosh:2003}. 
Hence we can define this model as a distribution on probability measures.

\begin{defn}[Adaptive P\'olya tree]
A probability measure $Q$ is said to have an {\em adaptive P\'olya tree} (APT) distribution with parameters $\bphi=(\bF_{\bnu},Q_0)$ if the corresponding PACs of $Q$, $\{\theta(A): A\in \A^{(\infty)}\}$, are generated from the above hierarchical model. We write $Q\sim {\rm APT}(\bF_{\bnu},Q_0)$, or equivalently ${\rm APT}(\bF_{\bnu},\bthe_0)$. 
\end{defn}
\noindent Remark: A hidden hyperparameter is the partition tree $\A^{(\infty)}$, which for simplicity we treat as given. One can also treat it as a parameter and even place a further layer of hyperprior on $\A^{(\infty)}$ as in \cite{hanson:2006}, resulting in a class of mixture of APTs.

\vspace{0.5em}

The meaning of $Q_0$ stays the same as for the PT---it is still the mean of the new model. 
\begin{thm}[Mean]
\label{thm:apt_prior_mean}
The mean of an APT distribution is $Q_0$. 
That is, for any Borel set $B\subset \om$, a random measure $Q$ that has the APT distribution satisfies $EQ(B)=Q_0(B)$.
\end{thm}

Under very general prior specifications, the APT model has full $L_1$ support and enjoys posterior consistency. We defer the study of these properties to Section~\ref{sec:markov_apt} where we derive these results for a more general model that contains the APT as a special case. The next theorem provides the Bayesian inference recipe for the APT.
\begin{thm}[Posterior]
\label{thm:posterior_apt}
Suppose $Q\sim {\rm APT}(\bF_{\bnu},Q_0)$. Given $n$ i.i.d.\ observations $\bx=(x_1,x_2,\ldots,x_n)$ from $Q$, the joint posterior of $(\bnu,\bthe)$ is given by 
\vspace{-3em}

\begin{align*}
\nu(A)\,|\,\bphi,\bx &\sim \tilde{F}_{\nu(A)}\\ 
\theta(A)\,|\,\bphi,\bm{\nu},\bx & \sim {\rm Beta}(\tilde{\theta}_0(A)\tilde{\nu}(A),(1-\tilde{\theta}_0(A))\tilde{\nu}(A)) 
\end{align*}
\vspace{-3em}

\noindent where $\tilde{\theta}_0(A)$ and $\tilde{\nu}(A)$ are as defined before, while $\tilde{F}_{\nu(A)}$ is given by
\vspace{-1em}

\[
d\tilde{F}_{\nu(A)}(\nu) = dF_{\nu(A)}(\nu)\cdot M_{A}(\bthe_0,\nu)/M_{A}(\bthe_0)
\]
\vspace{-2.5em}

\noindent with 
\vspace{-4em}

\begin{align*}
M_{A}(\bthe_0,\nu)&=\frac{\Gamma(\theta_0(A)\nu+n(A_l))\Gamma((1-\theta_0(A))\nu+n(A_r))\Gamma(\nu)}{\Gamma(\nu+n(A))\Gamma(\theta_0(A)\nu)\Gamma((1-\theta_0(A))\nu)}
\end{align*}
\vspace{-2.5em}

\noindent and
\vspace{-4.5em}

\begin{align*}
M_{A}(\bthe_0)&=\int M_{A}(\bthe_0,\nu)dF_{\nu(A)}(\nu).
\end{align*}
\end{thm}
\vspace{-0.7em}

\noindent Remark: $M_{A}(\bthe_0,\nu)$ is the marginal likelihood of the local binomial experiment on $A$ given $\bthe_0$---in particular $\theta_0(A)$---and $\nu(A)=\nu$, while $M_{A}(\bthe_0)$ is the marginal likelihood of the local binomial experiment given $\theta_0(A)$. Note that when $\nu=\infty$, $M_A(\bthe_0,\infty):=\lim_{\nu\uparrow \infty}M_{A}(\bthe_0,\nu)=\theta_0(A)^{n(A_l)}(1-\theta_0(A))^{n(A_r)}$, which is still the likelihood of the local binomial experiment.
\vspace{0.5em}

Because $\nu$ is one-dimensional, it is easy to evaluate $M_{A}(\bthe_0)$ numerically.  Specifically, we evaluate $f_{\nu(A)}=d F_{\nu(A)}/d\mu$ on a grid of different $\nu$ values $\nu_{(1)},\nu_{(2)},\ldots,\nu_{(H)}$ covering the support of $F_{\nu(A)}$, and approximate $M_{A}(\bthe_0)$ using a finite Riemann integral. This approximation becomes particularly straightforward when $F_{\nu(A)}$ is uniform on an interval $[a,b]$ under some transformed scale of $\nu$, such as $\log_{10}(\nu)$. In this case one can choose the grid points to be of equal distance on the transformed scale and $M_{A}(\bthe_0) \approx \frac{1}{H}\sum_{h=1}^{H} M_{A}(\bthe_0,\nu_{(h)})$.

\vspace{-0.5em}

\subsection{Stochastically increasing shrinkage and the Markov-APT}
\label{sec:markov_apt}
\vspace{-0.5em}

Under the APT, the $\nu(A)$'s are mutually independent both {\em a priori} and {\em a posteriori}. The shrinkage level for each PAC is inferred based on just the empirical evidence from the corresponding local binomial experiment. 
Our next objective is to allow borrowing of information across the binomial experiments in determining the shrinkage level for each. 

What information can and should be borrowed across the experiments? The answer very much depends on the inference problem at hand. 
In density estimation, a reasonable assumption adopted (explicitly or implicitly) in all statistical methods is the smoothness of the underlying density. Indeed, even ``jumpy'' densities---those with sharp changes---must be assumed to eventually smooth out at high enough resolution (or at a small enough ``bandwidth'') as opposed to infinitely oscillating in arbitrarily small regions, because otherwise reliable estimation is infeasible. In the multi-resolution framework, smoothness translates into an increase in shrinkage for higher resolutions. This is exactly the motivation for the increasing Beta parameters in the PT \citep{lavine:1992}.

One approach to incorporating increasing shrinkage with resolution in the APT is to choose the prior $F_{\nu(A)}$ such that it is supported on larger values for nodes in deeper levels of the partition tree. For example, one may let $F_{\nu(A)}$ be supported on $[l(k),u(k))]$ for $A\in\A^{k}$ while $l(k)$ and $u(k)$ are fixed increasing functions in $k$. 
However, this strategy will inherit the very same issues of the standard PT as illustrated in the Example~\ref{ex:motive}. First, the proper rate at which shrinkage should increase---in terms of the choice of $l(k)$ and $u(k)$---is typically unknown {\em a priori}. Second, in many applications, the smoothness of the underlying density is {\em heterogeneous} across the sample space---sharp boundaries or sudden jumps may lie in an otherwise smooth density. Such features are particularly common in applications involving anomaly or change-point detection, in which the jumps and boundaries are the main focus of inference. To capture such structures, one need the rate at which shrinkage increases with the resolution to vary across the sample space, which is again impossible to specify {\em a priori}.

Fortunately, one can infer the proper (varying) rate of increasing shrinkage from the data  What we need is a {\em stochastic} model for the increasing rate of shrinkage. To this end, we first introduce a latent mixture representation for $F_{\nu(A)}$. Specifically, we specify $F_{\nu(A)}$ using a mixture of $I$ component distributions in a monotone increasing stochastic order 
\vspace{-1.7em}

\[ F^{1}_{\nu(A)}\prec F^{2}_{\nu(A)}\prec \ldots \prec F^{I}_{\nu(A)}.\] 
\vspace{-2.7em}

\noindent In particular, we choose these components to have non-overlapping supports. For example $F^{i}_{\nu(A)}$ may be 
supported on an interval $(a(i),b(i)]$ where $0\leq a(1)< b(1) \leq a(2)< b(2) \ldots \leq a(I) < b(I) \leq\infty$
is a sequence of increasing support boundary points. 
For now we shall treat the number of components $I$ and each $F^{i}_{\nu(A)}$ as given, but will provide guidelines on choosing them in Section~\ref{sec:prior_spec}. We let $\bF_{\nu(A)}=(F^{1}_{\nu(A)},F^{2}_{\nu(A)},\ldots,F^{I}_{\nu(A)})$, and use $\bF_{\bnu}=\{\bF_{\nu(A)}: A\in\A^{(\infty)}\}$ to denote the totality of all component distributions. 

Moreover, we introduce a latent state variable $C(A)$ for each $A\in\A^{(\infty)}$ that indicates the mixture component $\nu(A)$ comes from: 
\vspace{-1.6em}

\[
\nu(A)\,|\,C(A)=i \,\, \sim\,\, F^{i}_{\nu(A)}.
\]
\vspace{-2.5em}

\noindent for $i=1,2,\ldots,I$. We refer to $C(A)$ as the {\em shrinkage state} on $A$, and let $\C=\{C(A):A\in\A^{(\infty)}\}$ be the collection of all shrinkage states.

Now we can enforce stochastically non-decreasing shrinkage along each branch of $\A^{(\infty)}$ by specifying a joint prior on $\C$ that prevents the shrinkage state from moving lower in any branch. 
That is, if $A_p$ is $A$'s parent in $\A^{(\infty)}$, then we require $C(A)\geq C(A_p)$. 
A simple stochastic model for $\C$ that can help us impose such a constraint is the Markov tree (MT) \citep{crouse:1998}, which links the $C(A)$'s using a Markov process such that the shrinkage state $C(A)$ depends on that of $A_p$ through Markov transition. 
The Markov process is initiated on the root, $\om$, as follows
\vspace{-1.5em}

\[
P(C(\om)=i) = \gamma_{i}(\om) \quad \text{for $i\in\{1,2,\ldots,I\}$},
\]
\vspace{-2.2em}

\noindent where the $\gamma_{i}(\om)$'s are called the {\em initial state probabilities}, and can be put into a vector
\vspace{-1.6em}

\[
\bgam(\om)=(\gamma_{1}(\om),\gamma_{2}(\om),\ldots,\gamma_{I}(\om))
\]
\vspace{-2.5em}

\noindent Then for each $A\neq \om$, $C(A)$ is determined sequentially based on its parent according to 
\vspace{-3.2em}

\begin{align*}
P(C(A)=i'\,|\,C(A_{p})=i)& =\gamma_{i,i'}(A) \quad \text{for $i,i'\in\{1,2,\ldots,I\}$} 
\end{align*}
\vspace{-3.2em}

\noindent where $\gamma_{i,i'}(A)$ is called the {\em state transition probability}, 
which can be organized into a {\em transition probability matrix}
\vspace{-2em}

\[
\bgam(A)=\left( \begin{array}{cccc}
\gamma_{1,1}(A) & \gamma_{1,2}(A) & \cdots & \gamma_{1,I}(A)\\
\gamma_{2,1}(A) & \gamma_{2,2}(A) & \cdots & \gamma_{2,I}(A)\\
\vdots & \vdots & \vdots & \vdots\\
\gamma_{I,1}(A) & \gamma_{I,2}(A) & \cdots & \gamma_{I,I}(A)
\end{array} \right).
\]

The desired stochastically increasing shrinkage is achieved when the transition matrices are all upper-triangular. That is,
\vspace{-1.4em}

\[
\gamma_{i,i'}(A) = 0 \quad \text{if $i>i'$ for all $A\in\A^{(\infty)}$.}
\]
From now on, we shall use $\bgam=\{\bgam(A):A\in\A^{(\infty)}\}$ to denote the collection of all initial state probabilities and transition probability matrices needed for specifying the MT. 

Putting the pieces together, now we have the following hierarchical model for a probability distribution
with hyperparameters $\bphi=(\bgam,\bF_{\bnu},Q_0)$
\vspace{-3em}

\begin{align*}
\C\,|\,\bphi &\sim {\rm MT}(\bgam)\\ 
\nu(A)\,|\, \bphi,\C &\sim \sum_{c=1}^{I} F^{i}_{\nu(A)}\cdot \I_{C(A)=i}\\ 
\theta(A)\,|\, \bphi, \bm{\nu},\C & \sim {\rm Beta}(\theta_0(A)\nu(A),(1-\theta_0(A))\nu(A)) 
\end{align*}
\vspace{-3em}

\noindent for all $A\in\A^{(\infty)}$. A graphical representation of this model is given in \ref{fig:graphical}(c). 
Because this hierarchical model also generates a probability measure with probability~1, one can again define it formally as a distribution on probability measures.
\begin{defn}[Markov adaptive P\'olya tree]
A probability measure $Q$ is said to have a {\em Markov adaptive P\'olya tree} (Markov\text{-}APT) distribution with parameters $\bphi=(\bgam,\bF_{\bnu},Q_0)$ if $Q$ corresponds to the collection of PACs $\{\theta(A): A\in \A^{(\infty)}\}$ generated from the above hierarchical model. We write $Q\sim {\rm Markov\text{-}APT}(\bgam,\bF_{\bnu},Q_0)$. In particular, when $\bgam(A)$ is upper-triangular for all $A\in\A^{(\infty)}$, we say that the Markov-APT is {\em stochastically increasing}.
\end{defn}
\vspace{-0.1em}

The meaning of $Q_0$ is still the same as before.
\vspace{-0.2em}

\begin{thm}[Mean]
\label{thm:mapt_prior_mean}
The mean of an Markov-APT is $Q_0$. That is, for any Borel set $B\subset \om$, a random measure $Q$ that has a Markov-APT distribution satisfies $EQ(B)=Q_0(B)$.
\end{thm}

In density estimation, we shall focus on stochastically increasing Markov-APTs. Moreover, in order to model a density, we must ensure that a random measure $Q$ generated from an Markov-APT has a density.
Earlier works in the literature have established two general approaches for achieving absolute continuity for PT-type priors, which could both be adopted for the Markov-APT. The first is to force $\nu(A)$ to increase with the level of $A$ at a sufficiently fast rate \citep{lavine:1992}. For the Markov-APT, this can be achieved by choosing $\bF_{\nu(A)}$ such that for any $A \in\A^{k}$, $\nu(A)>l(k)$ with probability~1  where $l(k)$ is a positive function in $k$ that satisfies $\sum_{k=1}^{\infty} 1/l(k) < \infty$ \citep{kraft:1964}. But this would impose a minimum amount of prespecified shrinkage homogeneously across the sample space, which is exactly the undesirable feature of the PT that we wish to avoid through adaptive shrinkage. For this reason we prefer an alternative strategy for ensuring absolute continuity \citep{wongandma:2010}, which is to include a separate ``complete shrinkage'' state as in the following theorem. 
\begin{thm}[Absolute continuity]
\label{thm:absolute_continuity}
Suppose $Q$ has a stochastically increasing Markov-APT distribution for which $F^{I}_{\nu(A)}=\I_{\infty}$, a point mass at $\infty$, and there exists $\delta>0$ such that for all large enough $k$, the state transition probabilities for any $A\in \A^{k}$ satisfies
\vspace{-1.8em}

\[
 \gamma_{i,I}(A) > \delta \quad \text{for $i=1,2,\ldots,I-1$},
\]
\vspace{-2.5em}

\noindent then with probability~1, $Q\ll Q_0$. In particular, if $Q_0 \ll \mu$, then $Q\ll \mu$.
\end{thm}
\noindent Remark: The complete shrinkage state eliminates the need for increasing lower bound on the support of $\nu(A)$ to ensure absolute continuity.

Next we establish two theoretical guarantees for inference using the Markov-APT model. The first property regards the flexibility of the model---it shows that Markov-APT enjoys full prior support. Thus inference with the Markov-APT is fully nonparametric.
\begin{thm}[Large prior support]
\label{thm:large_support}
Suppose $Q\sim \text{Markov-APT}(\bgam,\bF_{\bnu},Q_0)$ that satisfies the conditions of Theorem~\ref{thm:absolute_continuity}. In addition, suppose (i) $Q_0\ll \mu$, (ii) $I\geq 2$, (iii) $\exists \delta'>0$ such that for all large enough $k$, $\gamma_{i,I}(A)<1-\delta'$ for all $A\in\A^{k}$ and $i=1,2,\ldots,I-1$, and (iv)
$\exists \epsilon>0$ and $N>0$ such that $F^{i}_{\nu(A)}\bigl((0,N]\bigr) > \epsilon$ for all $i=1,2,\ldots,I-1$ and all $A\in \A^{(\infty)}$. Then for any distribution $G \ll Q_0$ and any $\tau>0$, we have
\vspace{-1em}

\[
\pi\left(Q:\int |q - g| d\,\mu < \tau\right) >0
\]
\vspace{-2em}

\noindent where $q=dQ/d\mu$ and $g=dG/d\mu$ are the corresponding densities with respect to $\mu$.
\end{thm}

 The next result regards the asymptotic consistency of inference using the Markov-APT and it guarantees that as we get more and more data, the posterior will eventually concentrate into any weak neighborhood of the true density. For any probability measures $P_0$ on $\om$, a {\em weak neighborhood} $U$ of $P_0$ is a set of probability measures on $\om$ of the form
\vspace{-0.8em}

\[
U = \Biggl\{ Q: \Big|\int f_i(\cdot) dQ - \int f_i(\cdot) d P_0\Big| < \epsilon_i, \text{ for $i=1,2,\ldots,K$} \Biggr\}
\]
\vspace{-1.7em}

\noindent for any bounded continuous functions $f_i$'s and non-negative constants $\epsilon_i$'s.

\begin{thm}[Posterior consistency under weak topology]
\label{thm:post_consistency}
Suppose $X_1,X_2,\ldots,X_n,\ldots$ are i.i.d.\ data from $Q$, and let $\pi(\cdot)$ be a stochastically increasing Markov-APT prior on $Q$ that satisfies the conditions in Theorem~\ref{thm:large_support} and let $\pi(\cdot|X_1,X_2,\ldots,X_n)$ be the corresponding posterior. Then for any $P_0\ll Q_0$ with bounded density $dP_0/dQ_0$, we have 
\vspace{-1.5em}

\[
\pi(U\,|\,X_1,X_2,\ldots,X_n) \longrightarrow 1 \quad \text{as $n\rightarrow \infty$}
\]
\vspace{-2em}

\noindent with $P_0^{(\infty)}$ probability~1 for any weak neighborhood $U$ of $P_0$. 
\end{thm}
\vspace{-1.7em}

\subsection{Bayesian inference with the Markov-APT} 
\label{sec:post_inf}
\vspace{-0.5em}

Next we address how to carry out posterior inference for the Markov-APT. 
We show that the full posterior can be derived following a general recipe for hierarchical models \citep[Sec.~5.3]{gelman:2013} and inference can proceed in a usual manner through drawing a sample from the posterior and/or computing some summary statistic such as the posterior mean. In particular, the posterior is available analytically and so no Markov Chain Monte Carlo (MCMC) is needed. In particular, the full posterior $\pi(\C,\bnu,\bthe\,|\,\bphi,\bx)$ can be described in three pieces: (i)~$\pi(\bthe\,|\,\bphi,\bnu,\C,\bx)$, (ii)~$\pi(\bnu\,|\,\bphi,\C,\bx)$, and (iii)~$\pi(\C\,|\,\bphi,\bx)$ as follows. 
\vspace{0.5em}

(i)~$\pi(\bthe\,|\,\bphi,\bnu,\C,\bx)$. This conditional posterior follows directly from the Beta-binomial conjugacy. Specifically, we have
\vspace{-3.2em}

\begin{align*}
\theta(A)\,|\,\bphi,\bm{\nu},\bx &\sim {\rm Beta}(\tilde{\theta}_0(A)\tilde{\nu}(A),(1-\tilde{\theta}_0(A))\tilde{\nu}(A)) \quad \text{for all $A\in\A^{(\infty)}$,}
\end{align*}
\vspace{-3.2em}

\noindent where as before $\tilde{\nu}(A)=\nu(A)+n(A)$ and $\tilde{\theta}(A)=\left(\theta_0(A)\nu(A)+n(A_l)\right)/\tilde{\nu}(A)$. 
\vspace{0.5em}

(ii)~$\pi(\bnu\,|\,\bphi,\C,\bx)$. 
Due to the conjugacy of finite mixture models, 
the conditional posterior for $\nu(A)$ is still an $I$-component mixture
\vspace{-1.5em}

\[
\nu(A)\,|\,\bphi,\C,\bx \sim \sum_{c=1}^{I} \tilde{F}^{i}_{\nu(A)} \cdot \I_{C(A)=i}
\]
\vspace{-2.2em}

\noindent where each new mixture component distribution $\tilde{F}^{i}_{\nu(A)}$ has density
\vspace{-1.6em}

\[
d\tilde{F}^{i}_{\nu(A)}(\nu) = d F^{i}_{\nu(A)}(\nu)\cdot M_{A}(\bthe_0,\nu)/M_{A}^{i}(\bthe_0)
\]
\vspace{-2.5em}

\noindent for $i=1,2,\ldots,I$. The normalizing constant (or partition function) in the above density
\vspace{-1.2em}

\[
M^{i}_{A}(\bthe_0)=\int M_{A}(\bthe_0,\nu) d F^{i}_{\nu(A)}(\nu)
\]
\vspace{-2em}

\noindent is the marginal likelihood of the local binomial experiment on $A$ given $\theta_0(A)$ and $C(A)=i$, which can be numerically evaluated as described following Theorem~\ref{thm:posterior_apt} for $M_A(\bthe_0)$.

\vspace{0.5em}

(iii)~$\pi(\C\,|\,\bphi,\bx)$. The last piece is the marginal posterior on $\C$, 
which follows again from the conjugacy of finite mixtures (of which the MT is a special case) to be an MT with initial and transition probabilities computable analytically through a forward-backward algorithm~\citep{liu:2001}. The forward step, or the summation step, is a bottom-up (leaf-to-root) recursion on the partition tree and the backward step, or the sampling step, a top-down recursion. 

To describe the algorithm, we first define a mapping $\xi_{A}$ for each $A\in\A^{(\infty)}$ as follows
\vspace{-1em}

\[
\xi_{A}(i,\bphi):=\begin{cases}
\int q(\bx | A)\pi(dq\,|\,\bphi,C(A_p)=i) & \text{if $A\in\A^{(\infty)}\backslash\{\om\}$}\\
\int q(\bx | A)\pi(dq\,|\,\bphi) & \text{if $A=\om$}
\end{cases}
\] 
for $i=1,2,\ldots,I$, where $q(\bx|A):=\prod_{x\in A} \frac{q(x)}{Q(A)}$ with $q=dQ/d\mu$, and $A_p$ is the parent of $A$ in $\A^{(\infty)}$. This mapping gives the marginal likelihood of the ``submodel'' on $A$---the Markov-APT with $A$ being the sample space---given that the shrinkage state on $A_p$ is $i$. Note that when $A=\om$, it does not have a parent, and $\xi_{\om}(i,\bphi)$ is equal for all $i$ to the overall marginal likelihood of the Markov-APT. The following lemma provides a bottom-up recursive recipe for computing $\xi_{A}(i,\bphi)$.

\begin{lem}[Forward-summation]
\label{lem:forward}
For $A\in\A^{(\infty)}\backslash\{\om\}$,
\vspace{-2.5em}

\[
\xi_{A}(i,\bphi)=\begin{cases}
 \sum_{i'=1}^{I}\gamma_{i,i'}(A)\cdot M^{i'}_{A}(\bthe_0)\cdot\xi_{A_l}(i',\bphi)\xi_{A_r}(i',\bphi) 
& \text{if $n(A)>1$}\\
q_0(\bx|A) & \text{if $n(A)=1$ or $A$ has no children}\\
1 & \text{if $n(A)=0$}
\end{cases}
\]
\vspace{-1.8em}

\noindent where $q_0(\bx|A)=\prod_{x_i\in A}q_0(x_i)/Q_0(A)$ with $q_0=dQ_0/d\mu$. For $A=\om$, we simply replace $\gamma_{i,i'}(A)$ with $\gamma_{i'}(\om)$ in the above equation. In particular, $\xi_{\om}(1,\bphi)$ is the overall marginal likelihood, as a function of the hyperparameters $\bphi$.
\end{lem}
\noindent Remark I: A node $A$ may have no children if $\om$ is discrete and $A$ is an atom.
\vspace{0.3em}

\noindent Remark II: This lemma shows that 
one can compute the mapping for $A$ based on those for its children, $A_l$ and~$A_r$ (hence bottom-up). Moreover, one can start the recursion from those $A$s such that $n(A)\leq 1$ but $n(A_p)\geq 2$, because all descendants of such $A$s have no more than one data point and so the mapping is known there.
\vspace{0.5em}

After computing $\{\xi_{A}(i,\bphi): A\in\A^{(\infty)}\text{ and }i=1,2,\ldots,I\}$, we can then carry out a backward (top-down) recursion to derive the marginal posterior of $\C$. 
\begin{thm}[Backward-sampling]
\label{thm:backward}
The marginal posterior of the shrinkage states is
\vspace{-1.5em}

\[
\C\,|\,\bphi,\bx \sim {\rm MT}(\bm{\tgam})
\]
\vspace{-2.5em}

\noindent whose initial state and transition probabilities $\bm{\tgam}=\{\bm{\tgam}(A):A\in\A^{(\infty)}\}$ are as follows.
\bi
\item The initial state probability vector:
\vspace{-2em}

\[
\bm{\tgam}(\om) = \bm{\gamma}(\om) \bm{D}''(\om)/\xi_{\om}(1,\bphi)
\]
\vspace{-3em}

\item The state transition probability matrix:
\vspace{-2em}

\[
\bm{\tgam}(A) = \bm{D}'(A)^{-1}\bm{\gamma}(A) \bm{D}''(A)
\]
\vspace{-3em}

\noindent for all $A\in \A^{(\infty)}\backslash \{\om\}$,
\ei
where for all $A\in\A^{(\infty)}$, $\bm{D}'(A)$ is the $I\times I$ diagonal matrix with the diagonal elements being $\xi_{A}(i,\bphi)$ for $i=1,2,\ldots,I$,
and $\bm{D}''(A)$ is the $I\times I$ diagonal matrix with the diagonal elements being $M^{i}_{A}(\bthe_0)\xi_{A_l}(i,\bphi)\xi_{A_r}(i,\bphi)$ if $A$ has children and $M^{i}_{A}(\bthe_0)$ if not for $i=1,2,\ldots,I$.
\end{thm}
\vspace{0.1em}

\noindent Remark: In particular, for any $A$ with $n(A)\leq 1$, by the theorem $\bm{\tgam}(A)=\bgam(A)$. So {\em a posteriori} the MT on $A$ with no more than one observation is the same as the prior MT. 

We have completely described the three components of the full posterior for $(\bthe,\bnu,\C)$. One can draw from the joint posterior by sampling in the order of $\pi(\C\,|\,\bphi,\bx)$, $\pi(\bnu\,|\,\bphi,\C,\bx)$, and $\pi(\bthe\,|\,\bphi,\bnu,\C,\bx)$. The forward-backward recursion given in Lemma~\ref{lem:forward} and Theorem~\ref{thm:backward} is the most computationally intense step in the posterior inference. Fortunately, for any given data set and prior specification, this recursion only needs to be carried out {\em once and for all}, because $\tilde{\bgam}$ stays the same for all the posterior draws.

After drawing $K$ posterior samples 
$(\bthe^{(1)},\bnu^{(1)},\C^{(1)}),(\bthe^{(2)},\bnu^{(2)},\C^{(2)}),\ldots,(\bthe^{(K)},\bnu^{(K)},\C^{(K)})$, one can carry out Bayesian inference in the usual manner. In particular, when one is interested in the unknown distribution $Q$, one can use $\bthe^{(1)},\bthe^{(2)},\ldots,\bthe^{(K)}$ to obtain a posterior sample $Q^{(1)},Q^{(2)},\ldots,Q^{(K)}$ while discarding the other variables. 

In fact, some posterior summaries can be evaluated analytically without resorting to posterior sampling at all. In particular, the PPD at any $x^{*}\in\om$ is equal to $\xi_{\om}^{*}(1,\bphi)/\xi_{\om}(1,\bphi)$, where $\xi_{\om}^{*}(1,\bphi)$ is the overall marginal likelihood computed according to Lemma~\ref{lem:forward} for a data set that contains the original data plus an additional point at $x^{*}$. This is particularly useful in applications such as density estimation because it avoids Monte Carlo errors in computing an estimator. In our numerical examples in \ref{sec:numerical_examples}, we compute all PPDs this way. Note that after computing the $\xi_{A}(i,\bphi)$ mappings for the original data set, the corresponding mappings $\xi_{\om}^{*}(i,\bphi)$ for the new data set can be very quickly obtained by updating only the branch of partition tree in which the new data point $x^{*}$ falls into, because the mapping stays the same on all other nodes. 

\vspace{-1em}

\subsection{Prior specification for Markov-APTs}
\label{sec:prior_spec}
\vspace{-0.5em}

Next we provide guidelines for specifying an Markov-APT, i.e.\ how to choose the hyperparameters $\{(\theta_0(A),\bgam(A),\bF_{\nu(A)}): A\in\A^{(\infty)}\}$, in the context of density estimation. The objective is to balance the dual-goal of robustness for a variety of distributional features, and parsimony---involving only a small number of hyperparameters to be set (i.e.\ the tuning parameters). We give an empirical Bayes strategy to set the tuning parameters adaptively. 
\vspace{0.5em}

{\em Prior choice of $\theta_0(A)$.} $\bthe_0$ is the PACs corresponding to $Q_0$, the prior mean of the model. That is, $\theta_0(A)=Q_0(A_l)/Q_0(A)$ for all $A\in\A^{(\infty)}$. Depending on the application, one may or may not have relevant prior knowledge regarding the prior mean. In lack of such knowledge, a simple default choice is to let $Q_0$ be uniform over a wide enough interval.  Another common situation is that one wants to center the Markov-APT around some parametric family such as the Gaussian location-scale family as in \cite{berger:2001,hanson:2006}, without specifying exactly which member of that family $Q_0$ is. This can be achieved by placing another layer of hierarchical prior on $Q_0$, e.g.\ on the location and scale parameters \citep{hanson:2006}, forming a mixture of Markov-APTs.
\vspace{0.5em}

{\em Prior choice of $\bgam(A)$.} In density estimation, we shall choose $\bgam(A)$ to be upper-triangular to achieve adaptive smoothing through stochastically increasing shrinkage. The most simple choice for the transition probabilities that satisfies this condition is
\vspace{-1em}

\[
\gamma_{i,i'}(A) = \begin{cases}
1/(I-i) & \text{if $i\leq i'$}\\
0 & \text{if $i > i'$}
\end{cases}
\]
for all $i,i'\in\{1,2,\ldots,I\}$ and $A\in\A^{(\infty)}\backslash\{\om\}$. In other words, given $A$'s parent is in shrinkage state $c$, then $A$ can take any higher shrinkage state (including $c$) with equal probability. 

This ``uniform'' transition probability specification is a special case of a more flexible kernel specification. Specifically, we can choose a kernel function $k(i,i')$ such that $k(i,i')$ is non-increasing in $|i-i'|$, and set the transition probability
\vspace{-1.5em}

\[
\gamma_{i,i'}(A) \propto k(i,i')\I_{i\leq i'}.
\]
\vspace{-2.2em}

\noindent A kernel $k(i,i')$ strictly decreasing in $|i-i'|$ will introduce ``stickiness'' into the shrinkage levels between a node and its parent (and thus also with its siblings and other relatives to a lesser degree). It encourages the shrinkage level to change gradually among nearby nodes in the partition tree, which is particularly useful when the smoothness (or lack of smoothness) of the underlying distribution tends to be similar for places closeby in the sample space. 

Of course, typically one does not know {\em a priori} whether and to what extent such sticky shrinakge is needed. Thus it is useful to allow the stickiness to be adaptively determined. One natural way to achieving this additional adaptivity is to choose a kernel that contains a tuning parameter for the stickiness, and use empirical Bayes to choose its value. For example, consider the exponential kernel
\vspace{-2.5em}

\[
k(i,i') = e^{-\beta|i-i'|},
\]
\vspace{-3em}

\noindent where $\beta\geq 0$ is the stickiness parameter. Note that $\beta=0$ corresponds to the uniform transition probabilities described above, while a large positive value of $\beta$ corresponds to strong stickiness in shrinkage. Finally, the initial state probabilities can be simply set to
\vspace{-1.5em}

\[
\gamma_{1}(\om) = \gamma_{2}(\om) =\cdots = \gamma_{I}(\om) = \frac{1}{I}. 
\]

{\em Prior choice of $\bF_{\bnu}$.}  Following a common practice \citep[Sec.~5.3]{gelman:2013} in Bayesian hierarchical modeling for Beta-binomial experiments, we specify prior on $\nu(A)$ on the log scale. First, we determine a {\em global} support for $\log_{10}\nu(A)$, i.e.\ the union of the supports of all $F_{\nu(A)}^i$. A convenient choice of the global support, aside from the complete shrinkage state included to ensure absolute continuity, is a finite interval $[L,U]$. 

A simple and robust strategy for choosing the interval that works in a wide variety of situations is to choose a wide enough range that covers all reasonable shrinkage levels and yet not so wide as to induce excessive prior probability in extremely strong or weak shrinkage levels. 
We recommend setting $[L,U]=[-1,4]$. On one end, $\log_{10}\nu(A)=-1$ corresponds to a prior sample size of $10^{-1}=0.1$, enforcing little shrinkage, while on the other $\log_{10}\nu(A)=4$ corresponds to shrinkage equivalent to about 10,000 prior ``observations'' for the local binomial experiment, resulting in very strong shrinkage. We have experimented with treating $L$ and $U$ as tuning parameters and choosing their values in a data dependent fashion using empirical Bayes (described below), but that resulted in little improvement over the very simple choice of [-1,4] in all of the numerical scenarios we investigated. 

Given the global support of $\log_{10}\nu(A)$, $[L,U]\cup\{\infty\}$, where $\infty$ is included for the complete shrinkage state, we now 
divide this support into $I$ non-overlapping intervals. Specifically, we let the first $(I-1)$ intervals evenly divide $[L,U]$ and the last being $\{\infty\}$. That is, we have
\vspace{-1.5em}

\[
[a(1),a(2)),\,\,\, [a(2),a(3)),\,\,\,\ldots,\,\,\, [a(I-1),a(I)),\,\,\, \{\infty\}.
\]
\vspace{-2.5em}

\noindent Then we let
\vspace{-2em}

\[
F^{i}_{\nu(A)}:\,\,\, \log_{10}\nu(A) \sim {\rm Uniform}(a(i),a(i+1))
\]
\vspace{-2.4em}

\noindent for $i=1,2,\ldots,I-1$ and $F^{I}_{\nu(A)}$ being a point mass at $\infty$. 
\vspace{0.5em}

{\em Choosing the tuning hyperparameters by empirical Bayes.} Our default prior specification is parsimonious in that it reduces the number of free parameters down to two---the number of shrinkage states $I$ and the stickiness parameter $\beta$ for the transition kernel. One can set these two tuning parameters in a data-adaptive manner by empirical Bayes. In particular, Lemma~\ref{lem:forward} provides the recipe for computing $\xi_{\om}(1,\bphi)$, the overall marginal likelihood. We can then compute the marginal likelihood as a function of the tuning parameters $\xi_{\om}(1,\bphi(I,\beta))$. Maximizing this likelihood over a grid of allowed values produces the maximum marginal likelihood estimate (MMLE) $(\hat{I},\hat{\beta})$, which one can then keep fixed in the inference.

\vspace{-1.2em}

\section{Performance evaluation in density estimation}
\label{sec:numerical_examples}
\vspace{-0.5em}

In this section we evaluate the performance of the Markov-APT in density estimation under five schematic simulation scenarios (\ref{fig:true_dens}). Each scenario corresponds to an underlying density with a particular type of structure commonly encountered in real applications. For each scenario, we simulate data sets of six different sample sizes---125, 250, 500, 750, 1000, and 1250. We compare the performance of Markov-APT to that of three other nonparametric models---namely, the PT, the OPT, and the DPM of normals \citep{escobar:1995}. We fit the DPM in {\tt R} using the library {\tt DPpackage} \citep{jara:2007,jara:2011b}. Details on the specification of the DPM are given in Supplementary Materials~S2.

For each method, we use the PPD, denoted by $\hat{f}$, as an estimator. To measure performance, we adopt the $L_1$ loss, i.e.\ the $L_1$ distance between $\hat{f}$ and the true density $f_0$, $||\hat{f}-f_0||_1 = \int |\hat{f}-f_0| d\mu$. For each simulation scenario and sample size, we generate $K=200$ data sets, with $\hat{f}^{(k)}$ being the estimate for the $k$th data set. We numerically calculate $||\hat{f}^{(k)} - f_0 ||_1$ using Riemann integral for all four methods. To make a comparison, for each of the competitors---PT, OPT, and DPM---we compute the percentage difference between its $L_1$ loss and that of the Markov-APT:
\vspace{-1em}

\[
\frac{||\hat{f}_{\text{Competitor}}^{(k)} - f_0 ||_1 - ||\hat{f}_{\text{Markov-APT}}^{(k)} - f_0 ||_1}{||\hat{f}_{\text{Markov-APT}}^{(k)} - f_0 ||_1} \times 100\%.
\]
A positive percentage increase indicates an outperformance of the Markov-APT over the competitor, with larger values indicating more effectiveness of the Markov-APT. Computing this measure of relative performance for each simulation allows us to evaluate both the average performance increase and the variability in the improvement across repeated experiments. In addition, we also estimate the average $L_1$ loss, i.e.\ the $L_1$ risk, $R_{L_1}(f_0,\hat{f})=E_{f_0} ||\hat{f}-f_0||_1$ for each method under each simulation setting using the Monte Carlo average
\vspace{-1.2em}

\[ 
\widehat{R}_{L_1}(f_0,\hat{f})= \frac{1}{K} \sum_{k=1}^{K} ||\hat{f}^{(k)} - f_0||_1.
\]

\vspace{-0.2em}

In all of the simulation settings, we adopt the prior specification recommended in Section~\ref{sec:prior_spec} with an exponential transition kernel, and use empirical Bayes to set the tuning parameters $(I,\beta)$. The range of tuning parameter values over which we maximize the marginal likelihood is $\{2,3,\ldots,11\}\times[0,2]$. The OPT also involves a tuning parameter $\rho_0$, the prior ``stopping'' probability \citep{wongandma:2010}, which we set by empirical Bayes using MMLE over $[0,1]$.
We implement the Markov-APT, PT, and OPT models up to the 12th level in the partition tree. Deeper partition trees result in little numerical difference.

\ref{fig:true_dens} presents the true densities for all scenarios from which the data are simulated. In each scenario, the underlying density is supported on the interval [0,1].  Our proposed model does not require the support to be a bounded interval, and this choice here is to simplify the numerical evaluation of the $L_1$ loss. This causes no loss of generality in our simulation because any density on $\real$ can be transformed onto [0,1] after applying, say, a cdf transformation. In \ref{fig:relative_L1} we present histograms of the percentage increase in $L_1$ loss for the three competitor methods relative to Markov-APT for each simulation setting. For easier comparison, we overlay the histograms for three different sample sizes---125 (small), 500 (medium), and 1000 (large)---to show how the relative performance changes for different sample sizes. (We have chosen to only show the three sample sizes in this figure rather than all six sample sizes because overlaying six histograms makes the plot illegible while using the three sample sizes is sufficient to convey the main finding.) \ref{fig:L1} presents the $L_1$ risks for all methods and scenarios versus sample size. Finally, to help understand why each method performs well or poorly in each scenario, in \ref{fig:pred_dens} we plot a typical PPD for each model under each scenario for a sample size that well differentiates the performance of the~methods. 

\bi
\item {\em Scenario 1:\ Spiky local structures.} In this case, the true distribution is
\vspace{-3.5em}

\[
\hspace{-2em} 0.2\, {\rm U}(0,1) + 0.2\,{\rm U}(0.2,0.205) + 0.2\,{\rm U}(0.4,0.405) + 0.2\,{\rm U}(0.6,0.605) + 0.2\, {\rm U}(0.8,0.805).
\]
\ei
\vspace{-1.5em}

\noindent See \ref{fig:true_dens}(a) for the true density. This represents the case when the underlying distribution has a few spiky structures in the midst of a flat background. In this case the key is to effectively determine the size (or height) of those spikes and pin down their boundaries.

Multi-resolution inference methods such as wavelets, PT, OPT, and our Markov-APT are particularly effective in capturing abrupt changes such as spikes and sharp boundaries in the nonparametric quantity of interest. Thus scenarios where the underlying distribution predominantly consists of spiky structures are the most favorable scenario for such methods in comparison to mixture-based method such as the DPM. 

Indeed, as shown in \ref{fig:relative_L1}(a) and \ref{fig:L1}(a), the Markov-APT 
performs substantially better than DPM, especially for medium and large sample sizes. It may first appear surprising that the PT, being a multi-resolution approach, performs the worst among all, in fact substantially worse than DPM. But this can be expected because the PT is unable to amply capture the height of the spikes due to over shrinkage at high-resolutions. In contrast, the amount of shrinkage under the Markov-APT is adaptive and thus automatically adjusts to low levels in and around the spikes. Interestingly, the OPT, which only allows no shrinkage or complete shrinkage, performs even slightly better than the Markov-APT. But in fact this is not surprising at all. When the underlying density is a step function as in the current scenario, the only appropriate shrinkage levels are indeed no shrinkage and complete shrinkage. Therefore the OPT is in fact an ``oracle'' in this case and one should expect it to perform the best. It is reassuring to see that the Markov-APT, while allowing much more flexible shrinkage levels, did not lose much efficiency relative to the ``oracle''.

From the PPDs in \ref{fig:pred_dens}(a) we see that the DPM tends to overestimate the height of the spikes---this is likely because in order to the characterize the sharp boundaries the DPM needs to ``squeeze'' the mixture component to have very thin tails, and thus making the mode of the mixture component much taller than the truth. 
\vspace{-0.5em}

\bi
\item {\em Scenario 2:\ Non-overlapping structures of different scales.} The true distribution is
\vspace{-1.7em}

\[
0.1\,{\rm U}(0,1) + 0.3\,{\rm U}(0.25,0.5) + 0.4\,{\rm Beta}_{(0.25,0.5)}(2,2) + 0.2\,{\rm Beta}(6000,4000)
\]
\ei
\vspace{-0.7em}

\noindent See \ref{fig:true_dens}(b) for the true density. This is the scenario given earlier in Example~\ref{ex:motive}. The underlying density has two bumps of different scales---one large and the other small---that are not overlapping with each other. The high-resolution, spiky structure now has a smooth boundary in contrast to the abrupt boundaries in the previous scenario. This is a favorable scenario for kernel mixture methods such as the DPM. By allowing the local kernel to have varying variance, the DPM is also able to adapt to the different scales of the two bumps. Thus one would expect the DPM to perform well.

From \ref{fig:relative_L1}(b) and \ref{fig:L1}(b), we see that the Markov-APT achieves better performance than the DPM at sample size 125, and comparable performance at larger sample sizes. The performance gain of Markov-APT over the OPT and PT is substantial at all sample sizes. \ref{fig:pred_dens}(b) shows that the OPT substantially oversmooths the large-scale feature while capturing the small-scale feature well, and the PT oversmooths the small-scale feature while undersmooths at high-resolutions within the large-scale feature.

\bi
\item {\em Scenario 3:\ Overlapping structures of different scales.} The true distribution is
\vspace{-2em}

\[
0.1\,{\rm U}(0,1) + 0.3\,{\rm U}(0.25,0.5) + 0.4\,{\rm Beta}_{(0.25,0.5)}(2,2) + 0.2\,{\rm Beta}(4000,6000).
\]
\ei
\vspace{-1em}

\noindent See \ref{fig:true_dens}(c) for the true density. This case is similar to the previous except that now the spiky local structure lies inside the smooth large-scale structure. From \ref{fig:L1}(c), we see that the performance of the Markov-APT is essentially unchanged from the case where the structures are non-overlapping. In contrast, \ref{fig:relative_L1}(c) and \ref{fig:L1}(c) show that the other methods all perform quite differently in this scenario. In particular, there is a substantial decay in performance for the DPM at smaller sample sizes compared to the case with non-overlapping structures, whereas the OPT and PT perform better for the current scenario. The dramatic change in the performance of PT compared to that in Scenario~2 illustrates the importance of adaptivity in achieving robust inference. Non-adaptive methods may be performing well in some situations but very poorly in another with only small or modest changes in the underlying distribution.

From \ref{fig:pred_dens}(c) we see that the three multi-resolution methods are all capable of characterizing the jump in the density even with small sample sizes. The sudden improvement in PT's performance is readily explained in \ref{fig:pred_dens}(c). The estimation error from PT in this and the previous scenario comes from two sources---the under-smoothing (i.e.\ under-shrinkage) in the large-scale smooth structure and over-smoothing (i.e.\ over-shrinkage) in the local spiky structure. By moving the spiky structure into the smooth structure, the error that comes from the under-smoothing in the large-scale structure is reduced because the smooth portion of the large-scale structure now accounts for a smaller proportion of the total probability mass, while the error that comes from over-smoothing the local structure is also reduced because now the local structure corresponds to more probability mass and hence more data. Note also in \ref{fig:relative_L1}(c) and \ref{fig:L1}(c) that the relative performance gain of Markov-APT over PT increases with the sample size.

\vspace{-0.5em}

\bi
\item {\em Scenario 4:\ Sharp boundaries.} The true distribution is
\vspace{-3.5em}

\[
\hspace{-2.5em} 0.1{\rm Beta}(2,2) + 0.25{\rm U}(0.3,0.55) + 0.05 {\rm Beta}_{(0.3,0.55)}(2,2) + 0.55 {\rm U}(0.55,0.8) + 0.05 {\rm Beta}(0.55,0.8).
\]
\ei
\vspace{-1.5em}

\noindent See \ref{fig:true_dens}(d) for the true density. In this scenario the underlying density is a couple of smooth structures with sharp boundaries separating them. In order to characterize the sharp jumps, the DPM introduces several mixture components, resulting in the loss of performance. \ref{fig:relative_L1}(d) and \ref{fig:L1}(d) show that the Markov-APT and the OPT perform the best and in very similar way over the entire range of sample sizes. From \ref{fig:pred_dens}(d) we see that the Markov-APT is able to both capture the sharp boundaries and the smooth modes, while the OPT again tends to oversmooth the two local modes, but is able to capture the sharp boundaries accurately. The $L_1$ loss in this example is predominantly contributed from the boundaries, and so the oversmoothing does not impair much of the performance of the OPT.  On the other hand, the PT again performs significantly worse than any other method due to its substantial undersmoothing. At small sample sizes, the DPM achieves comparable performance as the Markov-APT, but the relative performance gain of the Markov-APT over the DPM is widened for larger sample sizes. From \ref{fig:pred_dens}(d) we see that in order to characterize the sharp boundaries, the DPM needs to introduce a number of additional mixture components, resulting in overfitting in regions away from the boundaries.

\vspace{-0.5em}

\bi
\item {\em Scenario 5:\ Globally smooth structure.} The true distribution is ${\rm Beta}(10,20)$.
\ei
\vspace{-0.5em}

\noindent See \ref{fig:true_dens}(e) for the true density. In this case the underlying density is an approximately Gaussian smooth distribution, and DPM with a Gaussian kernel is essentially the true model for this scenario, and unsurprisingly performs the best. Such globally smooth distributions are the ``least favorable'' scenario for the three PT-type multi-resolution methods. \ref{fig:relative_L1}(e) and \ref{fig:L1}(e) show that for all sample sizes, the $L_1$ loss is on average about 50\% smaller under the DPM vs the Markov-APT. Among the multi-resolution methods, the Markov-APT substantially outperforms the OPT and the PT, and the performance gain widens for larger sample sizes.

\begin{figure}[p]
  \begin{center}
    \mbox{
      \subfigure[Scenario~1]{\includegraphics[width=0.35\textwidth]{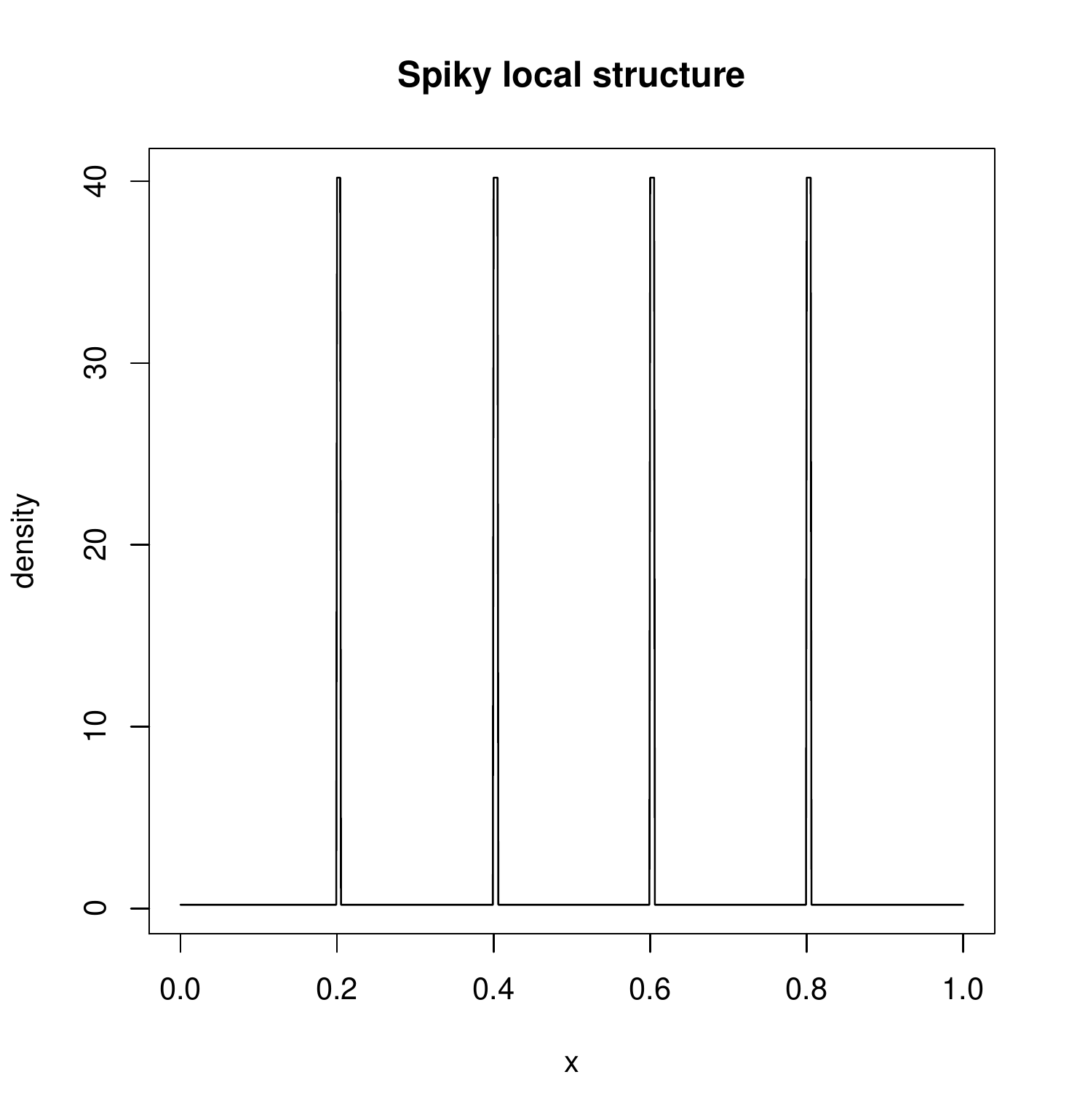}}
      \subfigure[Scenario~2]{\includegraphics[width=0.35\textwidth]{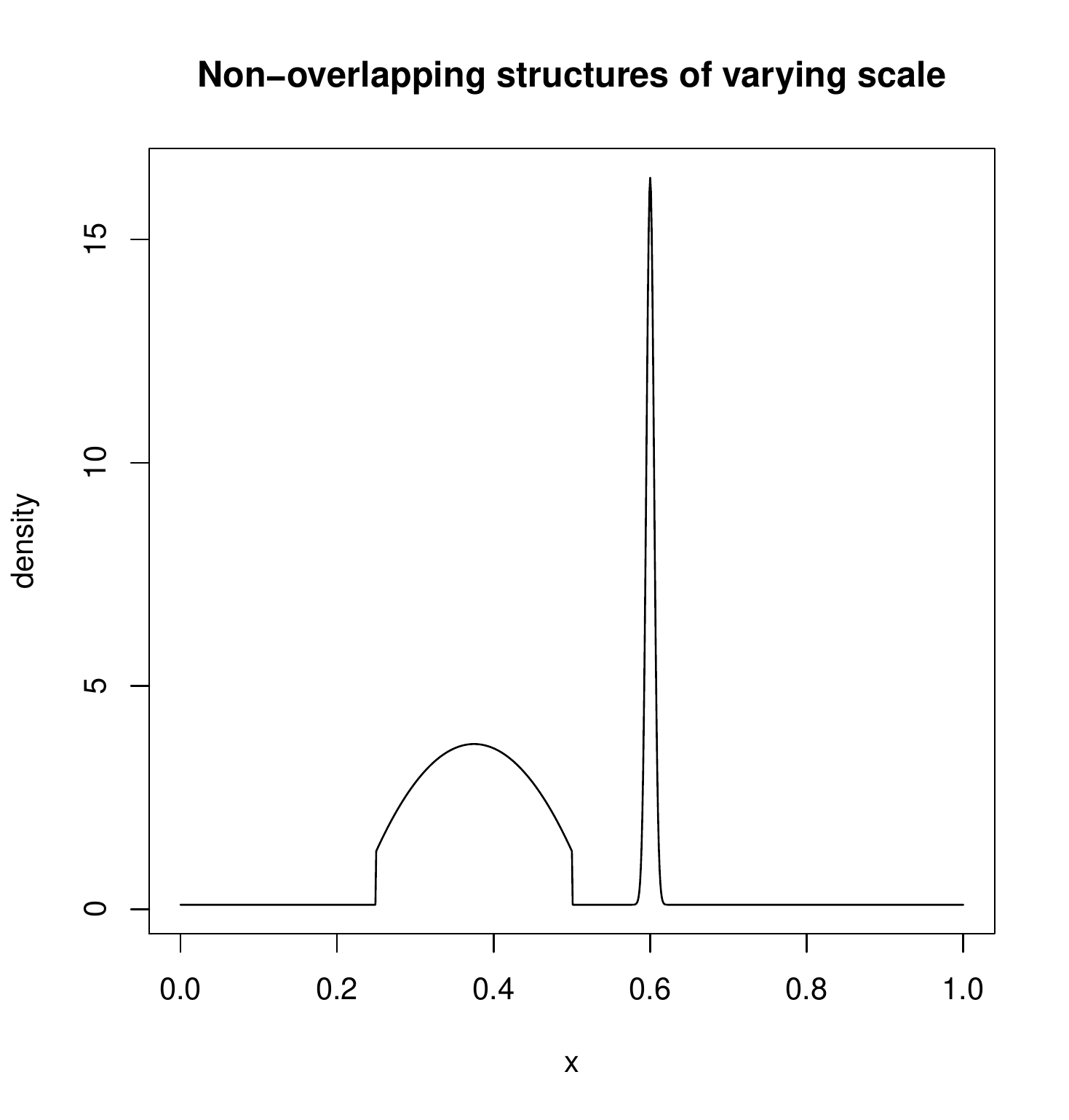}}
   }

    \mbox{
      \subfigure[Scenario~3]{\includegraphics[width=0.35\textwidth]{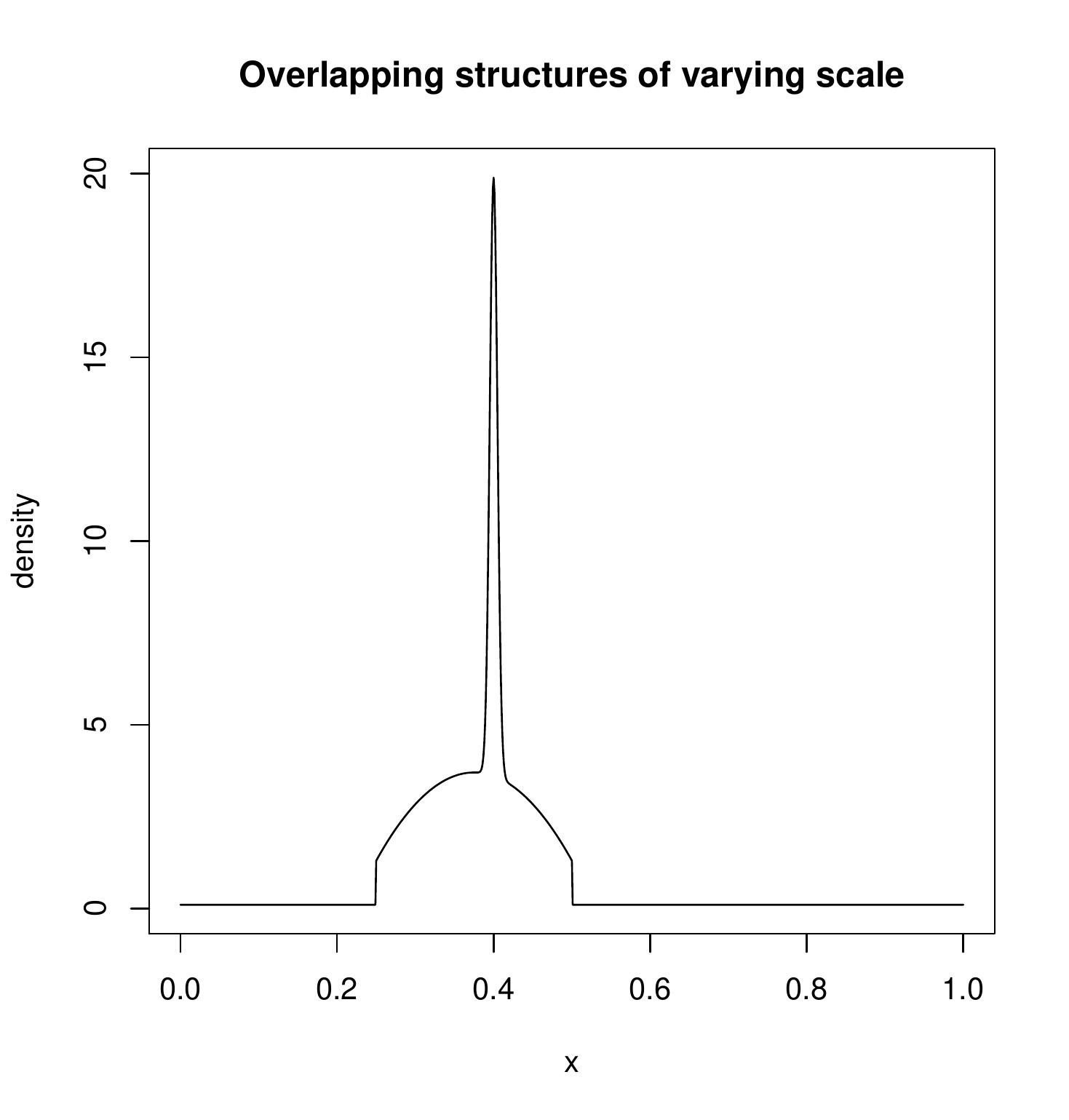}}
      \subfigure[Scenario~4]{\includegraphics[width=0.35\textwidth]{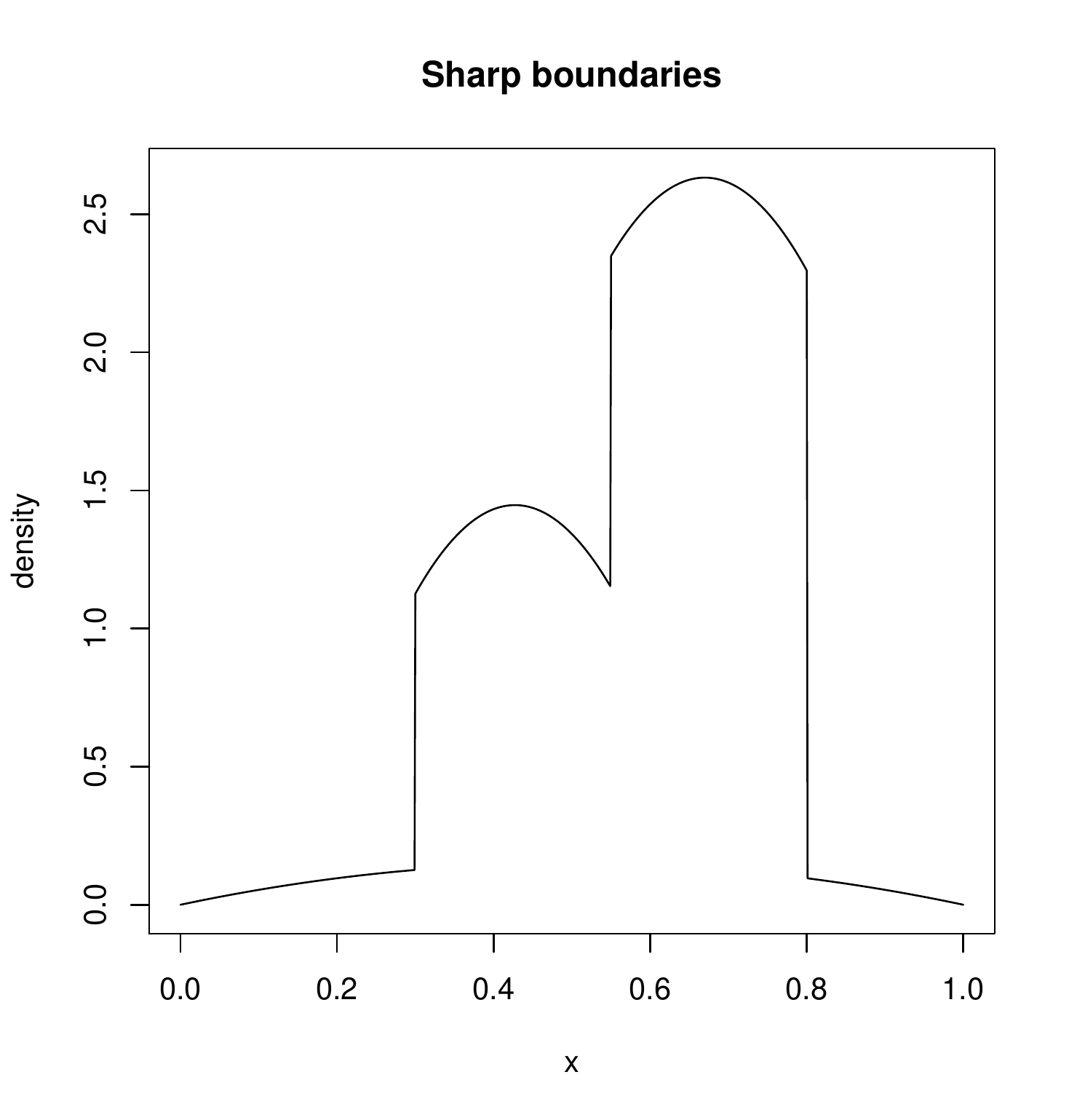}}

    }

      \subfigure[Scenario~5]{\includegraphics[width=0.35\textwidth]{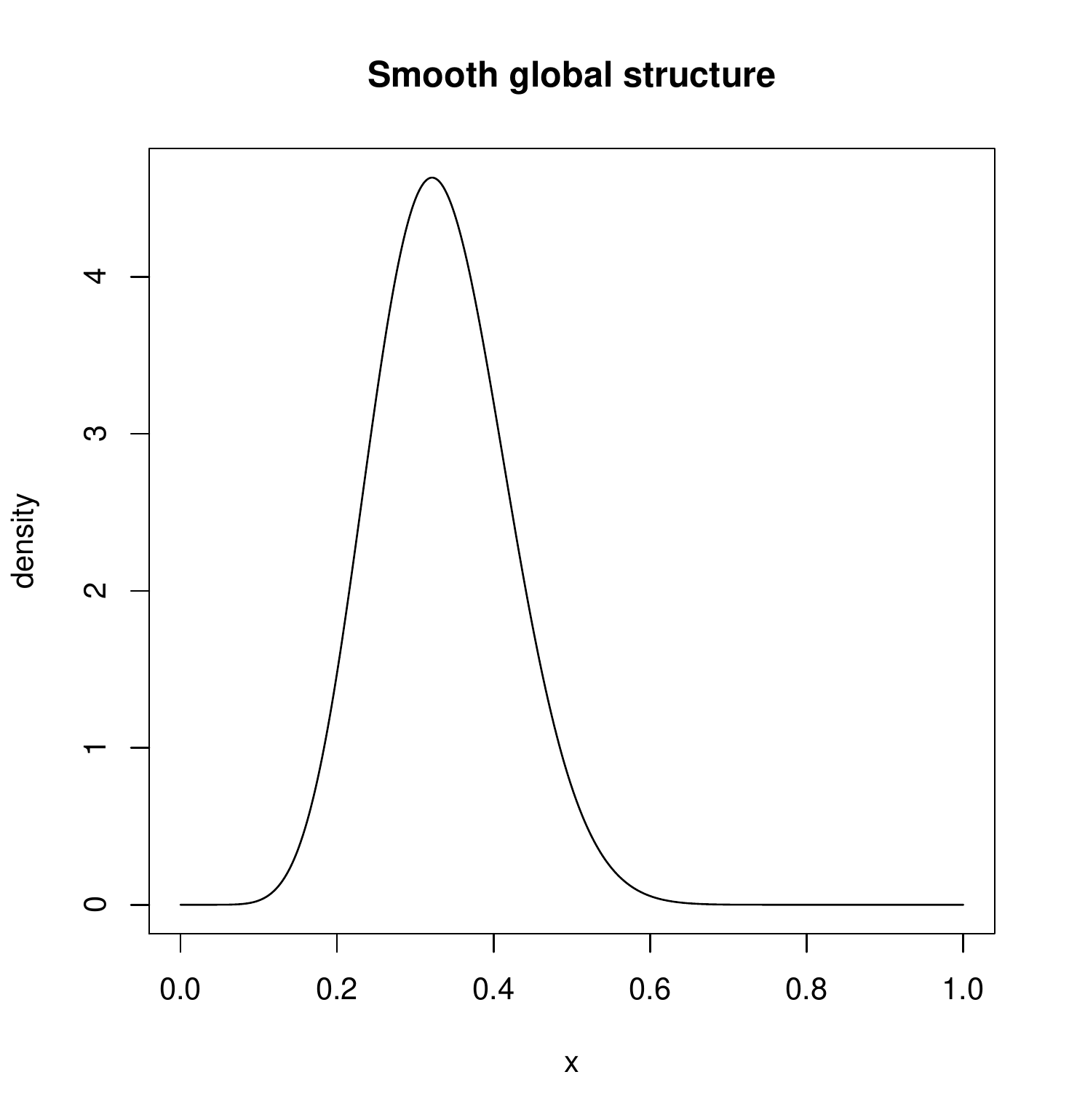}}
    \mbox{

   }
    \caption{True densities of the five simulation scenarios}
    \label{fig:true_dens}
  \end{center}
\end{figure}

\begin{figure}[p]
  \begin{center}
    \mbox{
      \subfigure[Scenario~1. Spiky local structure.]{\includegraphics[width=0.8\textwidth]{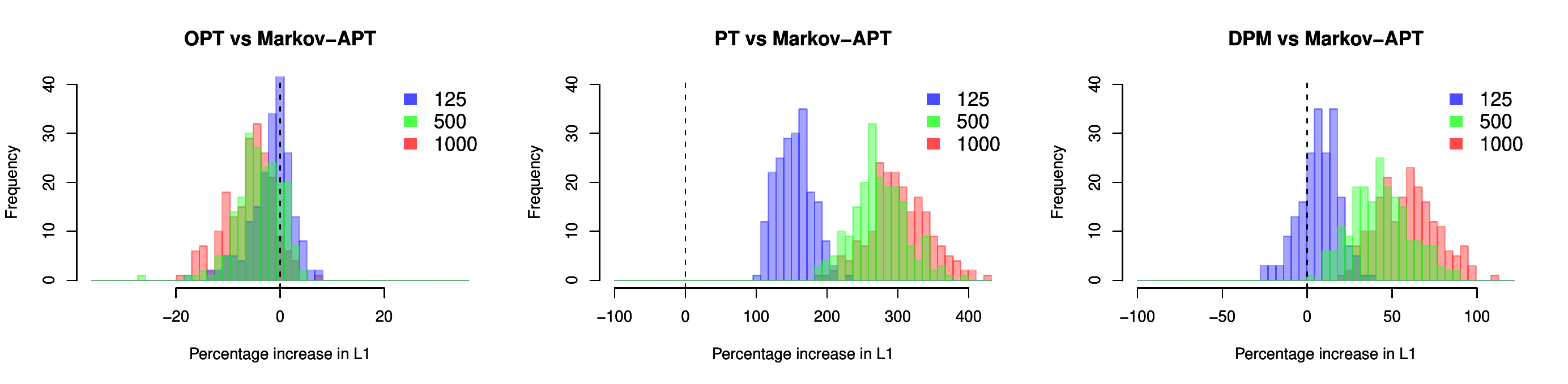}}
    }
    \mbox{
      \subfigure[Scenario~2. Non-overlapping structures of varying scale.]{\includegraphics[width=0.8\textwidth]{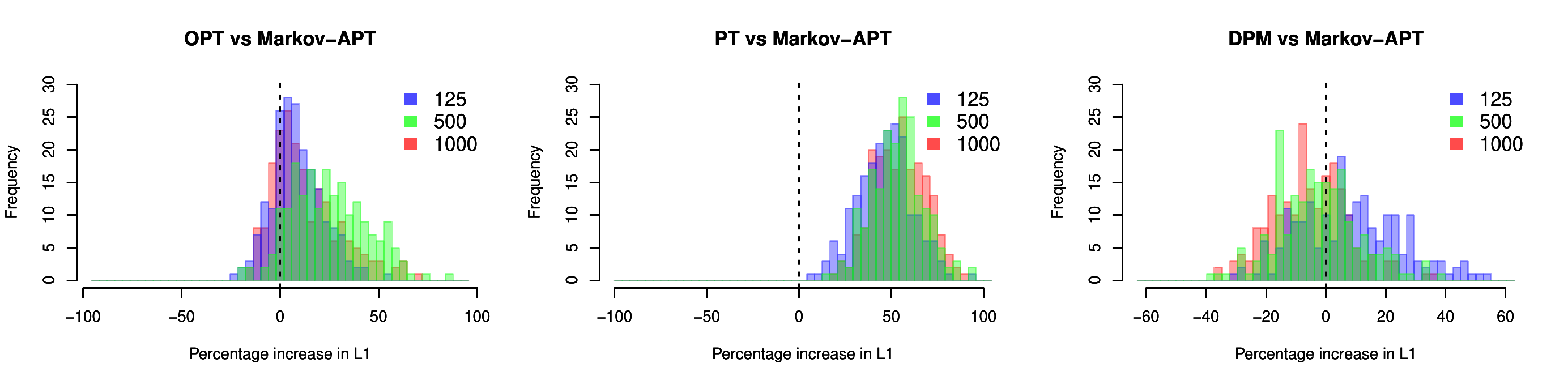}}
    }
    \mbox{
      \subfigure[Scenario~3. Overlapping structures of varying scale.]{\includegraphics[width=0.8\textwidth]{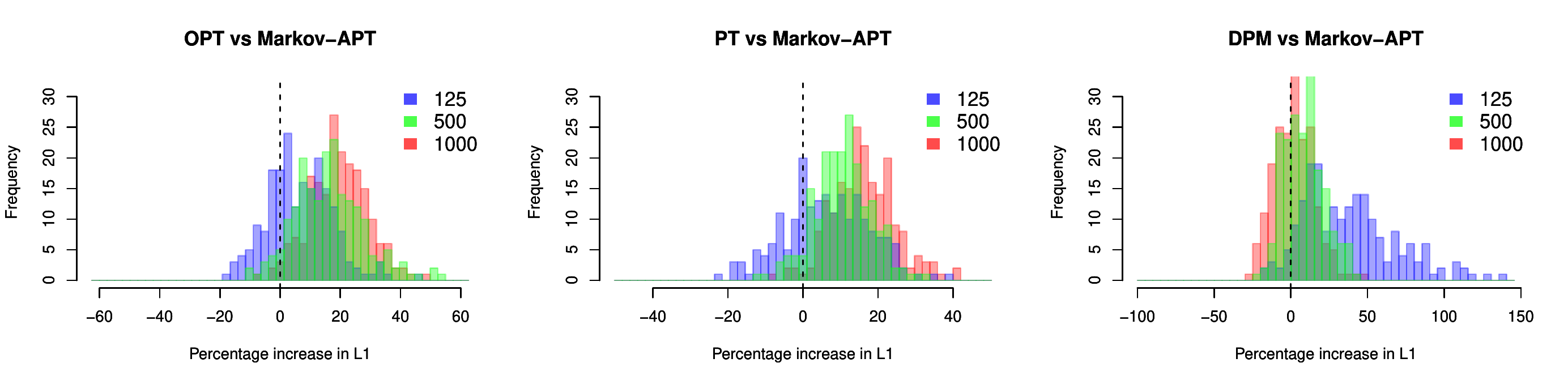}}
    }
    \mbox{
      \subfigure[Scenario~4. Sharp boundaries.]{\includegraphics[width=0.8\textwidth]{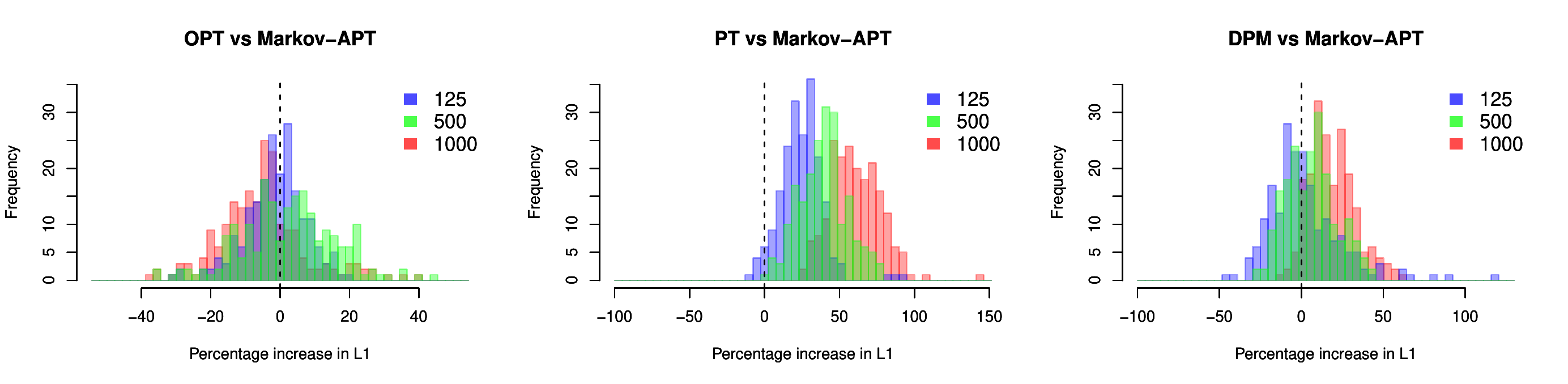}}
    }
    \mbox{
      \subfigure[Scenario~5. Smooth global structure.]{\includegraphics[width=0.8\textwidth]{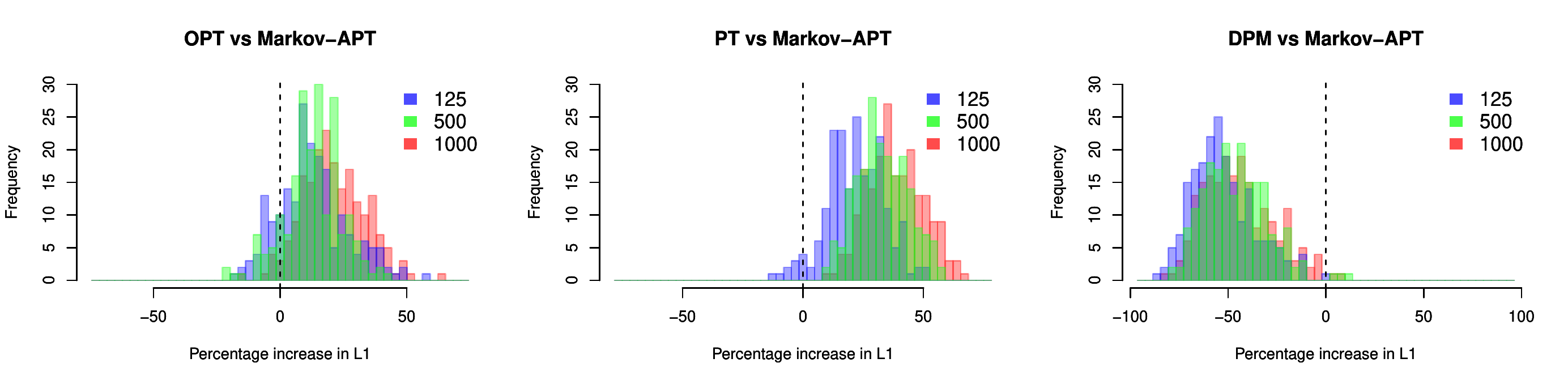}}
    }
    \caption{Histograms of percentage increase in $L_1$ risk for three methods compared to Markov-APT over three sample sizes---125, 500, and 1000. }
    \label{fig:relative_L1}
  \end{center}
\end{figure}

\begin{figure}[p]
  \begin{center}
    \mbox{
      \subfigure[Scenario~1]{\includegraphics[width=0.4\textwidth]{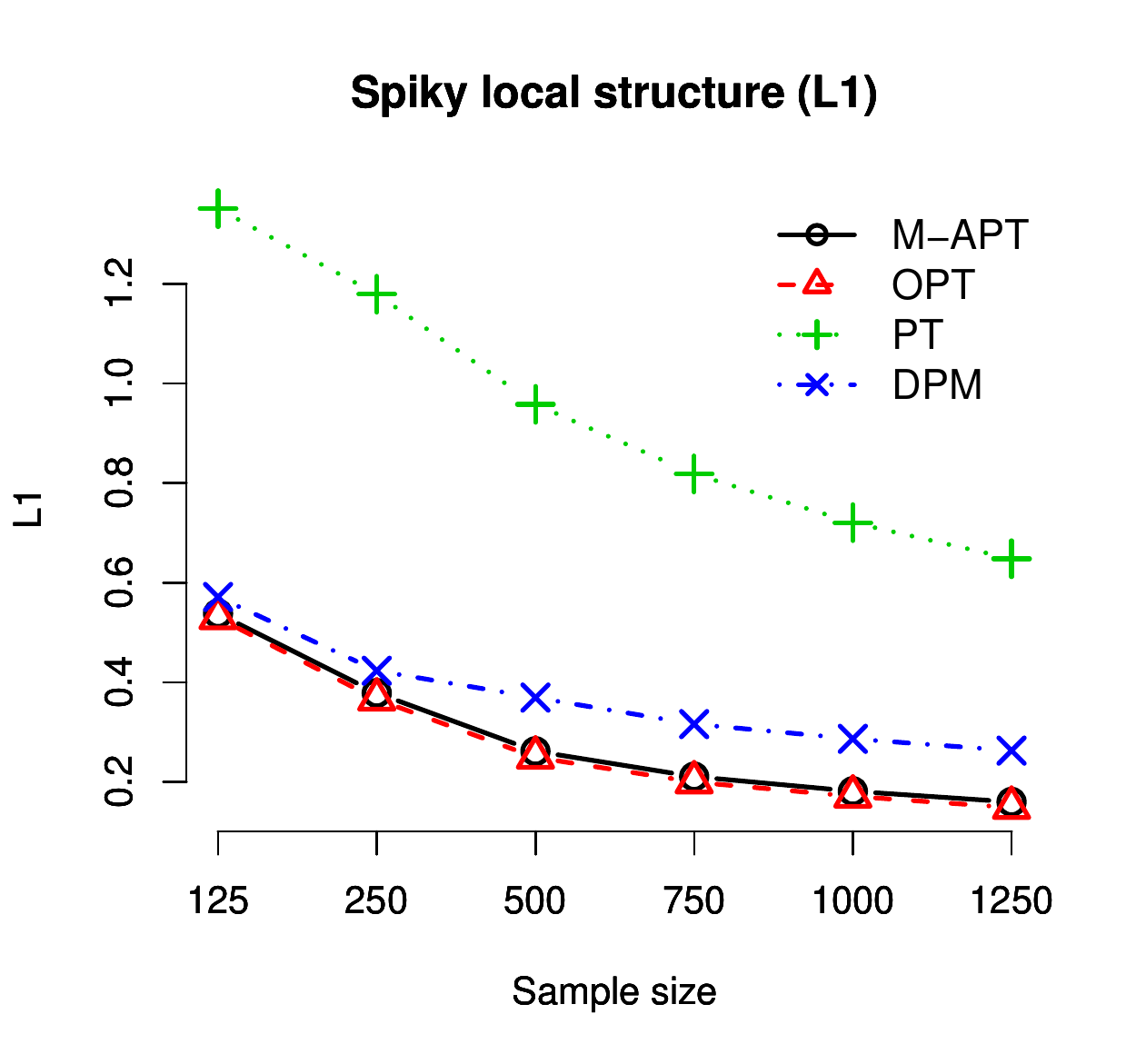}}
      \subfigure[Scenario~2]{\includegraphics[width=0.4\textwidth]{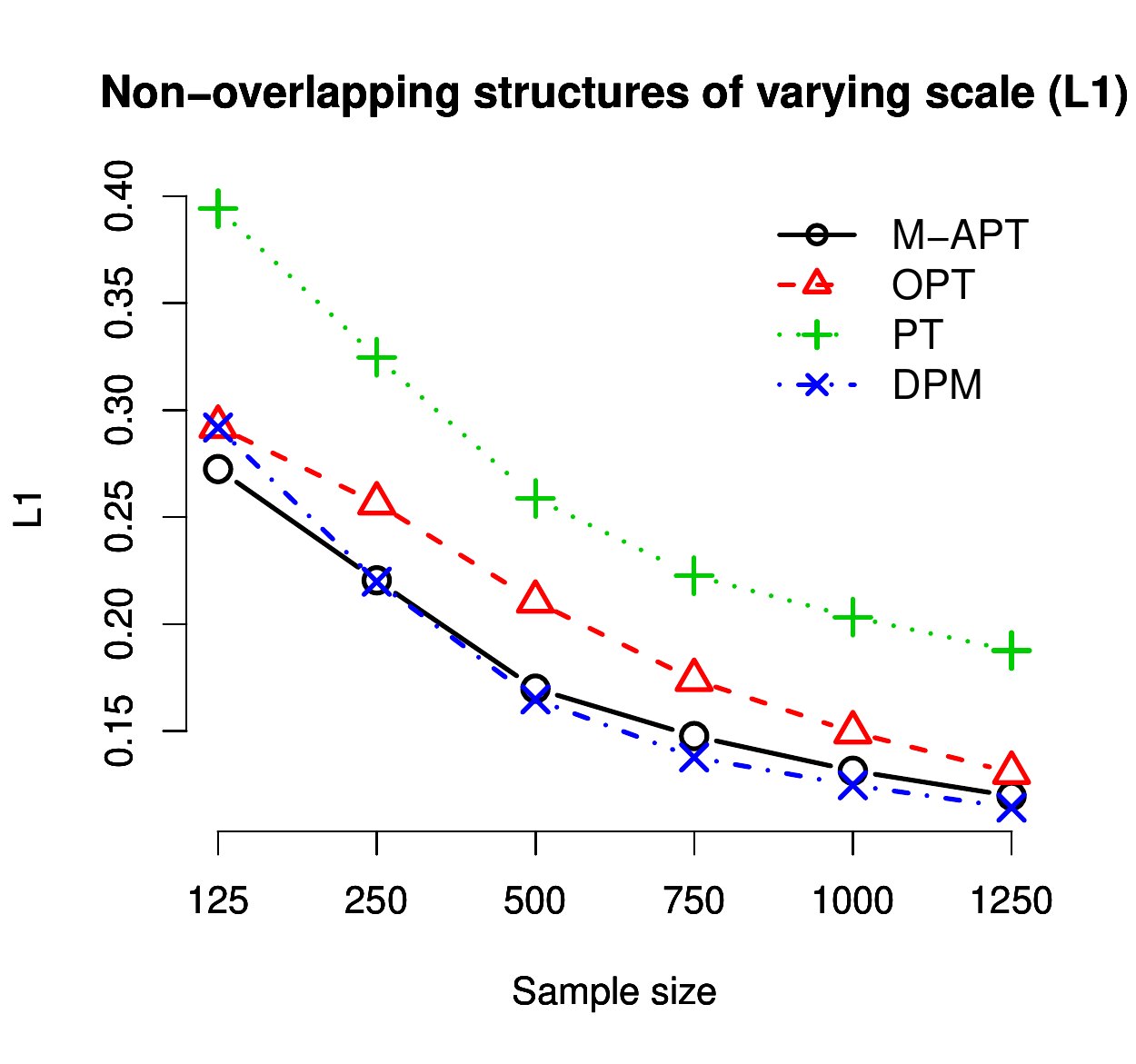}}
   }

    \mbox{
      \subfigure[Scenario~3]{\includegraphics[width=0.4\textwidth]{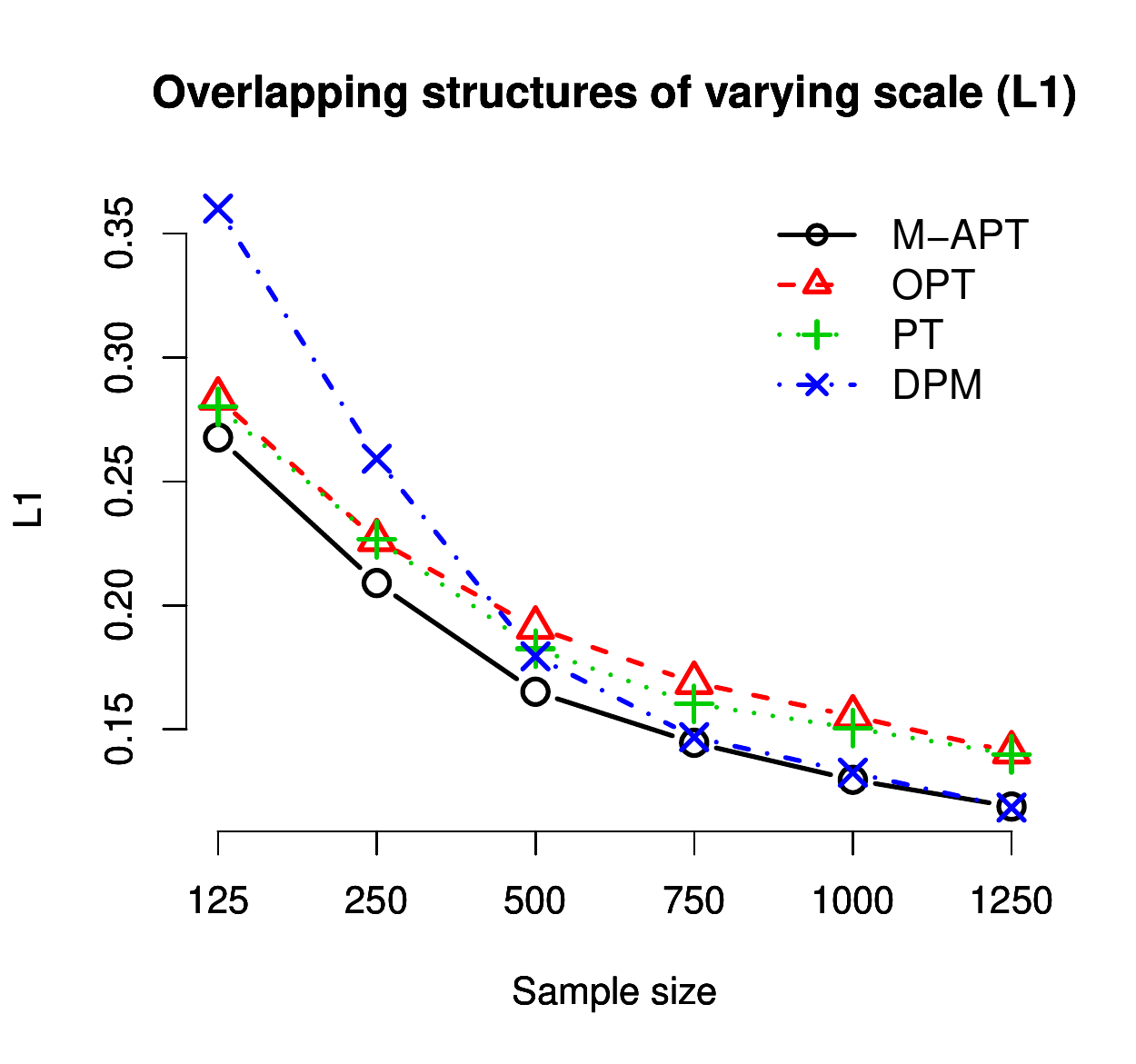}}
      \subfigure[Scenario~4]{\includegraphics[width=0.4\textwidth]{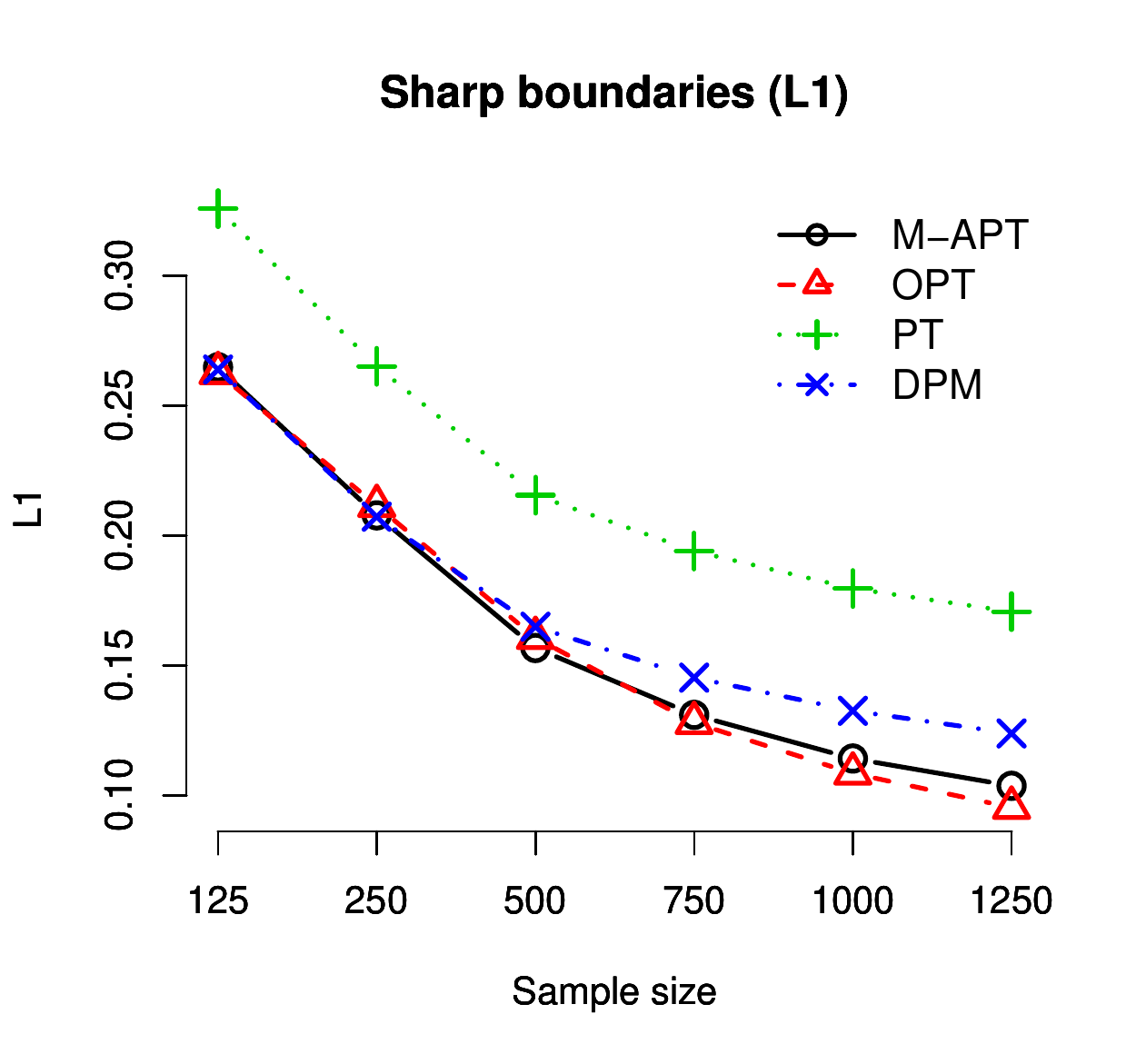}}

    }

    \mbox{
      \subfigure[Scenario~5]{\includegraphics[width=0.4\textwidth]{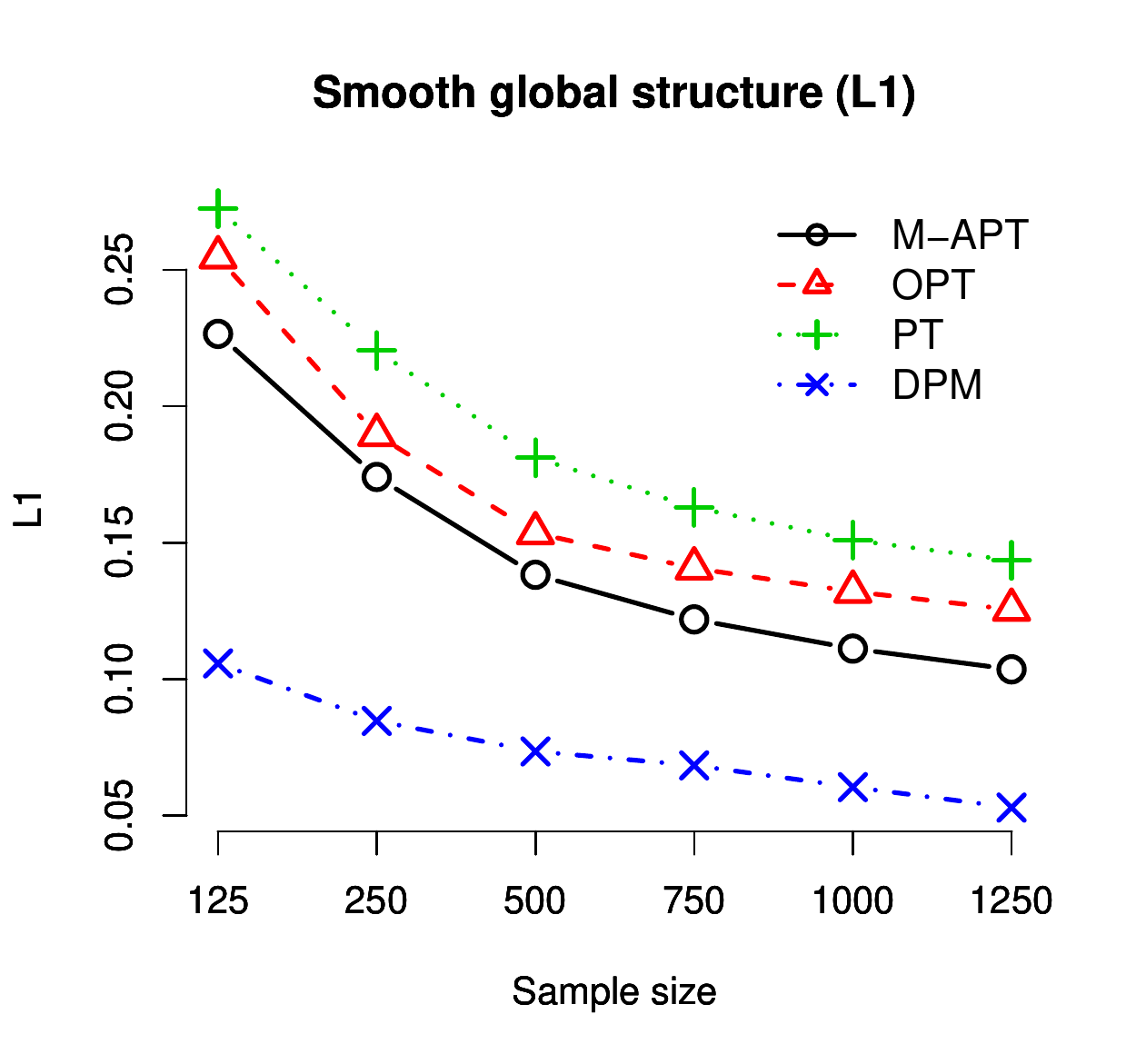}}

   }
    \caption{Estimated $L_1$ risk of four methods by sample size for the five simulation scenarios.} 
    \label{fig:L1}
  \end{center}
\end{figure}

\begin{figure}[p]
  \begin{center}
    \mbox{
      \subfigure[Scenario~1. $n=500$. Markov-APT: $\hat{I}=4$ and $\hat{\beta}=1.2$. OPT: $\hat{\rho}=0.26$.]{\includegraphics[width=0.8\textwidth]{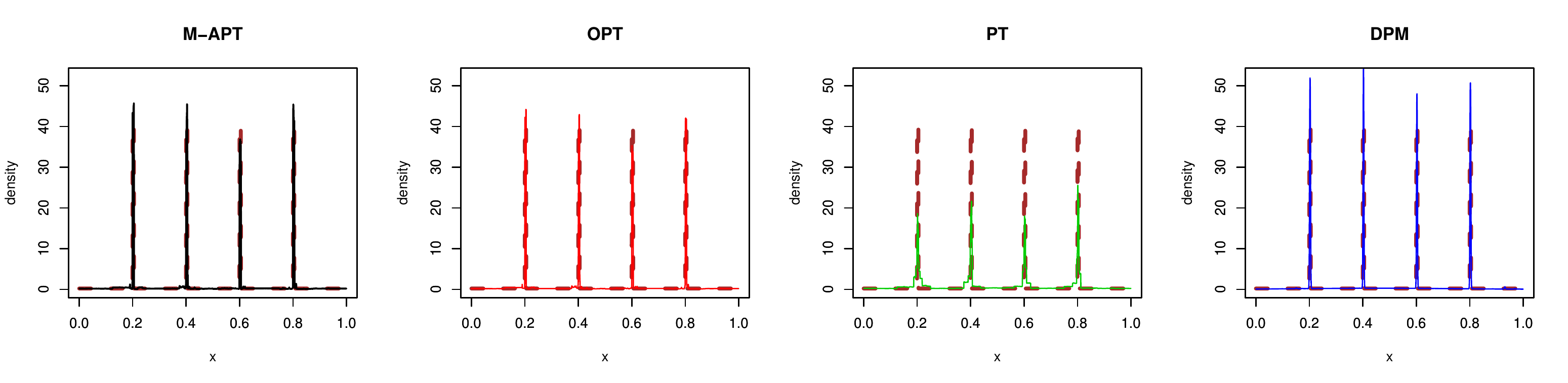}}
    }
    \mbox{
      \subfigure[Scenario~2. $n=500$. Markov-APT: $\hat{I}=8$ and $\hat{\beta}=0.70$. OPT: $\hat{\rho}=0.38$.]{\includegraphics[width=0.8\textwidth]{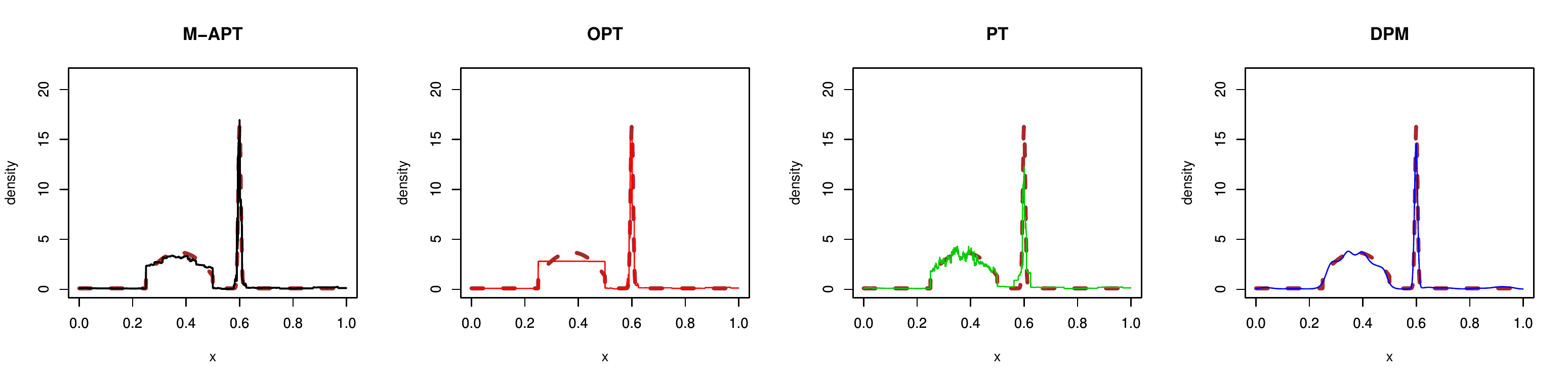}}
    }
    \mbox{
      \subfigure[Scenario~3. $n=500$. Markov-APT: $\hat{I}=11$ and $\hat{\beta}=0.50$. OPT: $\hat{\rho}=0.50$.]{\includegraphics[width=0.8\textwidth]{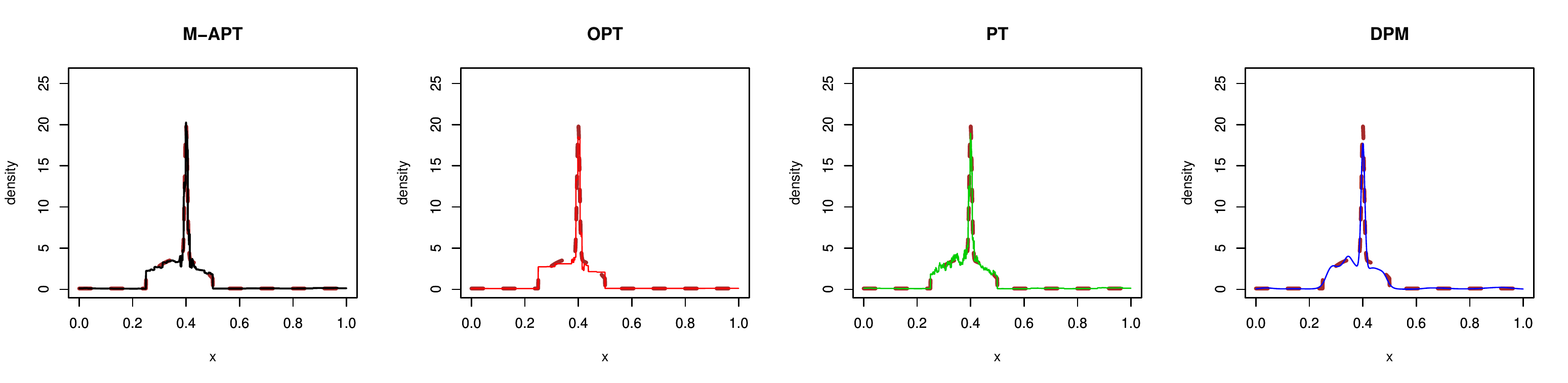}}
    }
    \mbox{
      \subfigure[Scenario~4. $n=1000$. Markov-APT: $\hat{I}=6$ and $\hat{\beta}=0.75$. OPT: $\hat{\rho}=0.42$.]{\includegraphics[width=0.8\textwidth]{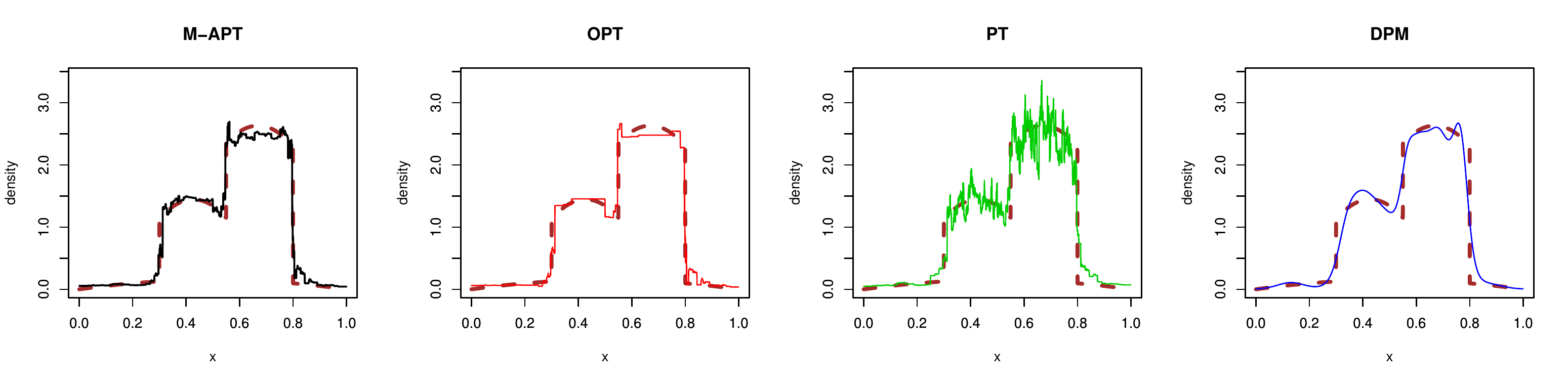}}
    }
    \mbox{
      \subfigure[Scenario~5. $n=500$. Markov-APT: $\hat{I}=11$ and $\hat{\beta}=0.60$. OPT: $\hat{\rho}=0.38$.]{\includegraphics[width=0.8\textwidth]{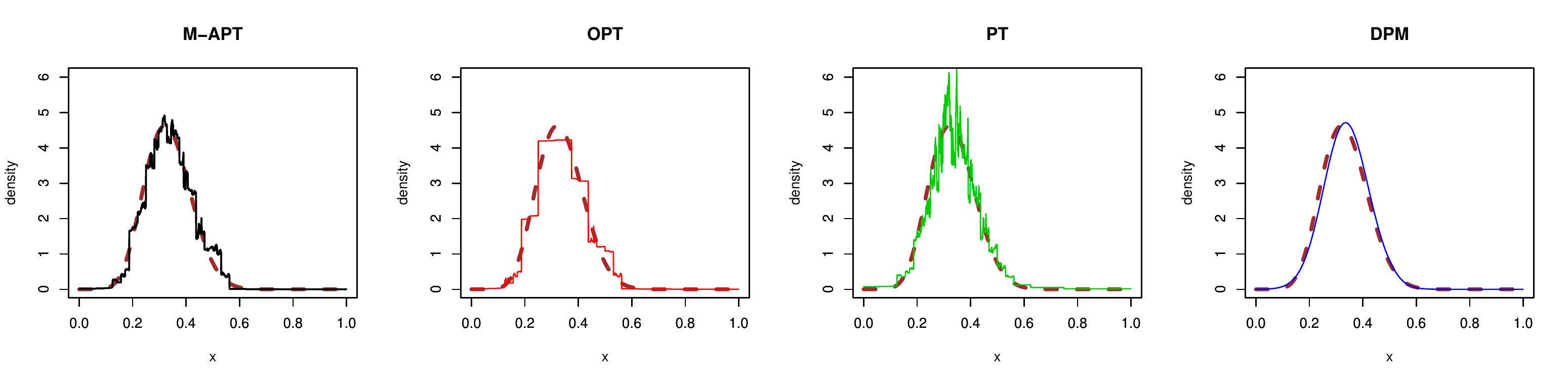}}
    }
    \caption{Typical PPDs (solid) of four methods and true density (dashed) for the five simulation scenarios. Sample sizes and tuning parameter values chosen by MMLE are given.}
    \label{fig:pred_dens}
  \end{center}
\end{figure}

\vspace{-1.5em}

\section{Discussion}
\label{sec:discussion}
\vspace{-0.7em}

We have showed that inference under PT-type multi-resolution models can be understood from a shrinkage perspective, and have introduced a hierarchical Bayesian approach to incorporating adaptive shrinkage. The APT and Markov-APT models can be easily applied in hypothesis testing such as for testing a parametric null versus a nonparametric alternative as previously studied in \cite{berger:2001} and testing two-sample differences as studied in \cite{holmes:2009,maandwong:2011,chen:2014}. Under this framework, the testing of features of the underlying distribution is transformed into testing a collection of local hypotheses organized on a partition tree, one for each node in the tree. As such, stochastically increasing shrinkage should be of less importance (as a motivation for choosing the appropriate probability transition matrix of the MT) than proper adjustment of multiple testing. Thus in such applications the strategy for prior specification is different. 

Last but very importantly, inference under PT-type multi-resolution models such as PT, OPT, APT and Markov-APT is extremely computationally efficient due to the conjugate hierarchical design and the forward-backward algorithm. 
In particular, for the most sophisticated model in this class, the Markov-APT, the dominating step in computing the posterior and PPD is computing the mappings $\xi_{A}(i,\bphi)$ through Lemma~\ref{lem:forward}, but it 
takes less than 0.1 second even with sample size 1250 for all of our numerical scenarios on a single Intel Core-i7 3.6Ghz CPU core with 400 Mbs of RAM. Moreoever, the computing time and required RAM stay essentially constant for sample sizes in the typical range (from tens to tens of thousands). In contrast, fitting the DPM model using MCMC is much more computationally expensive both in time and in memory. In particular, fitting the DPM in {\tt R} using the {\tt DPpackage} for each of the five scenarios at sample size 1250 takes about 4 to 7 minutes on the same machine and requires about 1.5 Gbs of RAM.

We believe that PT-type multi-resolution methods has tremendous potential for applications where the underlying distribution involves abrupt changes such as spikes or sharp boundaries, as well as where computational efficiency is of critical importance, such as in real-time change-point detection, online applications, and those applications with very large sample sizes. Therefore, additional effort is worthwhile to study the theory and further improve the statistical and computational performance of this class of methods.

\section*{Acknowledgment}
This research is supported by NSF grant DMS-1309057.

\vspace{-0.5em}

\bibliography{apt}

\newpage

\section*{Supplementary Materials}
\subsection*{S1.~Technical proofs}
\begin{proof}[Proof of Lemma~\ref{lem:uniqueness_of_PACs}]
The existence of a collection of PACs for $G$ follows immediately from the definition of PACs by letting $\theta(A)=G(A_l)/G(A)$ for each $A$ such that $G(A)>0$ and $\theta(A)=0$ otherwise. The uniqueness follows because $\A^{(\infty)}$ forms a $\pi$-system that generates the Borel $\sigma$-algebra. So by the extension theorem, two distributions with the same PACs on all $A$s such that $G(A)>0$ must be the same up to a set of $\mu$-measure 0.
\end{proof}

\begin{proof}[Proof of Theorem~\ref{thm:apt_prior_mean}]
Given $\bnu$, $Q$ has a PT distribution with mean $Q_0$. That is $E(Q(B)|\bnu)=Q_0(B)$. The result follows immediately by the law of iterated expectation.
\end{proof}

\begin{proof}[Proof of Theorem~\ref{thm:posterior_apt}]
This theorem follows immediately from applying the Bayes rule to the conjugate hierarchical model, and is described in detail in Section~\ref{sec:post_inf}. 
\end{proof}

\begin{proof}[Proof of Theorem~\ref{thm:mapt_prior_mean}]
Given $\C$, $Q$ has an APT distribution with mean $Q_0$ by Theorem~\ref{thm:apt_prior_mean}. That is $E(Q(B)|\C)=Q_0(B)$. The result follows immediately by the law of iterated expectation.
\end{proof}

\begin{proof}[Proof of Theorem~\ref{thm:absolute_continuity}]
Let $Q^{(k)}$ be the level-$k$ truncated version of $Q$. That is, $Q^{(k)}(A)=Q(A)$ for all $A\in\A^{(k)}$ and $Q^{(k)}(\cdot|A)=Q_0(\cdot|A)$ for all $A\in \A^{k}$. By the same argument as in the proof of Theorem~1 in \cite{wongandma:2010} (with $Q_0$ replacing $\mu$), we know that $Q^{(k)}$ converges in total variational distance to $Q$ as $k\rightarrow \infty$. By construction, $Q^{(k)}\ll Q_0$ for all $k$. Now for any set $B$ such that $Q(B)>0$, then there must exist some $k$ such that $Q^{(k)}(B)>0$, and therefore $Q_0(B)>0$. Hence $Q\ll Q_0$.
\end{proof}

\begin{proof}[Proof of Theorem~\ref{thm:large_support}]
Let $\tilde{q}=q/q_0$ and $\tilde{g}=g/q_0$ where $q_0=dQ_0/d\mu$. Our goal is to prove that for any $\tau > 0$,
\[
P\left(\int|\tilde{q}-\tilde{g}| dQ_0 < \tau\right) > 0.
\]
First we assume that $\tilde{g}$ is continuous and bounded, and let $M$ be a finite upperbound of $\tilde{g}$. For any $\sigma>0$, there exists a compact set $E$ such that there is a partition $\om= \cup_{i} A_{i}$ such that the diameter of each $A_{i}\cap E$ is less than $\sigma$. By the absolute continuity of $G$ w.r.t $Q_0$, there exists $\beta(\sigma)>0$ such that $G(E^c)<\beta(\sigma)$ if $Q_0(E^c)<\sigma$ and $\beta(\sigma)\downarrow 0$ as $\sigma\downarrow 0$. We define the modulus of continuity of $\tilde{g}$ on $E$ as
\[
\delta_{E}(\epsilon) = \sup_{x,y\in E: |x-y| < \epsilon} |\tilde{g}(x)-\tilde{g}(y)|.
\]
Note that by the continuity of $\tilde{g}$ and the compactness of $E$, $\delta_{E}(\epsilon)\downarrow 0$ as $\epsilon \downarrow 0$. Now we approximate $\tilde{g}$ by a step function $\tilde{g}^{*}(x)=\sum_{i} \tilde{g}_{i}^{*} I_{A_{i}}$ where $\tilde{g}_{i}^{*} = \int_{A_{i}\cap E} \tilde{g} dQ_0/Q_0(A_{i}\cap E)$. Let $D_{\epsilon}(\tilde{g})$ be the set of step functions $h(\cdot)=\sum_{i} h_i I_{A_{i}}(\cdot)$ such that $\sup_{i}|h_i-\tilde{g}_{i}^{*}|< \delta_{E}(\epsilon)+M\sigma$.

Suppose $h\in D_{\epsilon}(\tilde{g})$. For any $B\in \B$, the Borel sets, we have $B_{i}=B\cap A_{i}$. Then
\begin{align*}
&\Big| \int_{B} (h-\tilde{g}) dQ_0 \Big| \leq \sum_{i} |h_i-\tilde{g}_{i}^{*}| Q_0(B_{i}) + \sum_{i}\Big| \tilde{g}_{i}^{*}\, Q_0(B_{i}) - \int_{B_{i}} \tilde{g}dQ_0 \Big| \\
\leq& (\delta_{E}(\epsilon)+M\sigma) Q_0(B) + \sum_{i}\Big| \tilde{g}_{i}^{*}\, Q_0(B_{i}\cap E) - \int_{B_{i}\cap E} \tilde{g}dQ_0 \Big| + \sum_{i}\Big| \tilde{g}_{i}^{*}\, Q_0(B_{i}\cap E^c) - \int_{B_{i}\cap E^c} \tilde{g}dQ_0 \Big|\\
\leq& (\delta_{E}(\epsilon)+M\sigma)Q_0(B) + \sum_{i} r_{i} + 2 M \cdot Q_0(E^c)\\
<& (\delta_{E}(\epsilon)+M\sigma) Q_0(B) + \sum_{i} r_{i} + 2M\sigma
\end{align*}
where
\begin{align*}
r_{i} &= Q_0(B_{i}\cap E) \Bigg|\frac{\int_{A_{i}\cap E} \tilde{g} dQ_0}{Q_0(A_{i}\cap E)} - \frac{\int_{B_{i}\cap E}\tilde{g} dQ_0}{Q_0(B_{i}\cap E)} \Bigg|\\
&= Q_0(B_{i}\cap E) \Bigg|\frac{\int_{A_{i}\cap E} \left(\tilde{g}(x)  - \tilde{g}(x_{i})\right)q_0(x)\, dx}{Q_0(A_{i}\cap E)} - \frac{\int_{B_{i}\cap E}\left(\tilde{g}(x) - \tilde{g}(x_i)\right) q_0(x)\, dx}{Q_0(B_{i}\cap E)} \Bigg|
\end{align*}
for some $x_{i} \in B_{i}$. Thus
\[
|r_{i}| < 2 \delta_{E}(\epsilon) Q_0(B_{i})
\]
and so
\[
\Big| \int_{B} (h-\tilde{g}) dQ_0 \Big| < 3 \delta_{E}(\epsilon) Q_0(B) + 3M\sigma\quad \text{for all $B\in \B$.}
\]
Therefore by taking $B=\{x:h>\tilde{g}\}$ and $B=\{x:h\leq \tilde{g}\}$, we get
\[
\int  | h - \tilde{g} | dQ_0 < 3 \delta_{E}(\epsilon) + 6M\sigma.
\]
Now we let $Q^{(k)}$ be the level-$k$ truncated version of $Q$. That is, $Q^{(k)}(A)=Q(A)$ for all $A\in\A^{(k)}$ and $Q^{(k)}(\cdot|A)=Q_0(\cdot|A)$ for all $A\in \A^{k}$.
By the conditions in the theorem, we have for $\tilde{q}^{(k)}= q^{(k)}/q_0$ where $q^{(k)}=d Q^{(k)}/d\mu$,
\[
P\left(\tilde{q}^{(k)} \in D_{\epsilon}(\tilde{g})\text{ for all large $k$}\right) > 0.
\]
Thus
\[
P\left(\int |\tilde{q}^{(k)}-\tilde{g}|d Q_0 < 3\delta_{E}(\epsilon)+6M\sigma  \text{ \,\,for all large $k$}\right) > 0.
\]
But since 
\[
P\left(\int |\tilde{q}^{(k)} - \tilde{q}| dQ_0 \rightarrow 0\right) = 1,
\]
combining these we get
\[
P\left(\int |\tilde{q}-\tilde{g}|d Q_0 < 4 \delta_{E}(\epsilon)+6 M\sigma \right) > 0.
\]
The result follows by letting $\epsilon \downarrow 0$ and $\sigma\downarrow 0$.

Finally, if $\tilde{g}$ is not continuous and bounded, then since $Q_0$ is a probability measure, $\tilde{g}$ can be approximately arbitrarily well in $L_1$ w.r.t $Q_0$ by a continuous bounded density.
\end{proof}

\begin{proof}[Proof of Theorem~\ref{thm:post_consistency}]
Let $p_0=dP_0/d\mu$, $q_0=dQ_0/d\mu$, $\tilde{p}_0=dP_0/dQ_0$, and for any $Q\ll Q_0$, $\tilde{q}=dQ/dQ_0$. Let $M$ be a finite upperbound of $\tilde{p}_0$. Then the Kullback-Leibler (K-L) distance between $p_0$ and $q$ is given by
\[
{\rm KL}_{\mu}(p_0,q) = \int p_0\log(p_0/q)d\mu = \int \tilde{p}_0 \log(\tilde{p}_0/\tilde{q}) dQ_0 = {\rm KL}_{Q_0}(\tilde{p}_0,\tilde{q}).
\]
By Lusin's theorem we have a compact $E\subset \om$ with $Q_0(E^c)<\epsilon'$ such that $\tilde{p}_0$ is continuous on $E$. This $E$ can be chosen such that for every $\epsilon>0$, there exists a partition, $\om=\cup_{i} A_i$ with all $A_i \in \A^{(k)}$ for some $k$, such that the diameter of each $A_i\cap E$ is less than $\epsilon$.
We define
\[
\delta_{E}(\epsilon) =  \sup_{x,y\in E :|x-y|<\epsilon} |\tilde{p}_0(x)-\tilde{p}_0(y)|
\quad \text{and} \quad
d_i = \max\left(\sup_{A_i\cap E} \tilde{p}_0(x) + \delta_{E}(\epsilon),\epsilon'\right)
\]
and let $D_{\epsilon}(\tilde{p}_0)$ be the collection of step functions $g(x)=\sum_{i}g_i \I_{A_i}(x)$ with $d_i \leq g_i < d_i + \delta_{E}(\epsilon)$.
For every $g\in D_{\epsilon}(\tilde{p}_0)$, let $\tilde{g}$ be the normalized version of $g$, that is $\tilde{g}=g/\int g dQ_0$. Then
\[
\int_{E} (g-\tilde{p}_0) d Q_0 -  \int_{E^c}|g-\tilde{p}_0| dQ_0  \leq \int (g-\tilde{p}_0) d Q_0 \leq \int_{E} (g-\tilde{p}_0) d Q_0 + \int_{E^c}|g-\tilde{p}_0| dQ_0,
\]
and so
\[
\delta_{E}(\epsilon) - (2M+\epsilon') \epsilon' \leq \int (g-\tilde{p}_0) d Q_0 \leq 3\delta_{E}(\epsilon) + (2M+\epsilon')\epsilon'.
\]
Thus for any fixed $\epsilon$, when $\epsilon'$ is small enough, we have $\int (g-\tilde{p}_0) d Q_0 \geq 0$, and thus,
\[
\log\left(\int g dQ_0\right) = \log\left(1+ \int (g-\tilde{p}_0) d Q_0 \right) \leq 3\delta_E(\epsilon) + (2M+\epsilon')\epsilon'.
\]
Now,
\begin{align*}
0 \leq {\rm KL}_{Q_0}(\tilde{p}_0,\tilde{g}) &= \int \tilde{p}_0\log(\tilde{p}_0/\tilde{g}) dQ_0\\
&=\int_{E}\tilde{p}_0\log(\tilde{p}_0/g)dQ_0 + \int_{E^c}\tilde{p}_0\log(\tilde{p}_0/g)dQ_0 + \log\left(\int g dQ_0\right)\\
&\leq M\log(M/\epsilon')\epsilon' + 3\delta_E(\epsilon) + (2M+\epsilon') \epsilon'.
\end{align*}
By first choosing $\epsilon'$ and then $\epsilon$ small enough, we can make ${\rm KL}_{Q_0}(\tilde{p}_0,\tilde{g})$ arbitrarily small. So $p_0$ lies in the K-L support of $\pi$. Therefore, by Schwartz's theorem, we have posterior consistency at $p_0$ under the weak topology.
\end{proof}

\begin{proof}[Proof of Lemma~\ref{lem:forward}]
If $n(A)=0$ then by definition $\xi_{A}(i,\bphi)=1$. If $A$ has no children, then also by definition $\xi_{A}(i,\bphi)=q_0(\bx|A)$. If $n(A)=1$, then $\xi_{A}(i,\bphi)$ becomes the conditional prior predictive density on $A$ valued at $x$, which is just $q_0(x|A)$ since the Markov-APT conditional $A$ is still an Markov-tree and by Theorem~\ref{thm:mapt_prior_mean} its predictive density is $Q_0$. Finally we consider the case when $A$ has children and $n(A)\geq 2$. For $A\in\A^{(\infty)}\backslash \om$,
\begin{align*}
\xi_{A}(i,\bphi) &= \int q(\bx | A)\pi(dq\,|\,\bphi,C(A_p)=i) \\
&= \int q(\bx | A)  \pi(dq\,|\,\bphi,C(A)=i')\gamma_{i,i'}(A)\\
&= \sum_{i'=1}^{I} \int q(\bx\,|\, A)  \pi(dq\,|\,\bphi,C(A)=i')\gamma_{i,i'}(A)\\
&= \sum_{i'=1}^{I} \gamma_{i,i'}(A) \int \theta(A)^{n(A_l)} (1-\theta(A))^{n(A_r)} \pi(\theta(A)\,|\,C(A)=i') \times \\
&\hspace{7em}  \int q(\bx\,|\,A_l)q(\bx\,|\,A_r) \pi(dq\,|\,\bphi,C(A)=i')\\
&= \sum_{i'=1}^{I} \gamma_{i,i'}(A) M_{A}^{i'}(\bthe_0) \xi_{A_l}(i',\bphi)\xi_{A_r}(i',\bphi).
\end{align*}
For $A=\om$, the derivation follows similarly with $\gamma_{i,i'}(A)$ replaced by $\gamma_{i'}(A)$.
\end{proof}

\begin{proof}[Proof of Theorem~\ref{thm:backward}]
This theorem follows by two applications of Bayes rule. For $A\in\A^{(\infty)}\backslash\om$, the posterior transition probability is
\begin{align*}
\tilde{\gamma}_{i,i'}(A) &= P(C(A)=i'\,|\,C(A_p)=i,\bx(A))\\
&=\int q(\bx\,|\,A) \gamma_{i,i'}(A) \pi(dq\,|\,C(A)=i')\big/\int q(\bx\,|\,A)\pi(dq\,|\,C(A_p)=i)\\
&=\begin{cases}
\gamma_{i,i'}(A)M_{A}^{i'}(\bthe_0) \xi_{A_l}(i',\bphi)\xi_{A_r}(i',\bphi)/\xi_{A}(i,\bphi) & \text{if $A$ has children}\\
\gamma_{i,i'}(A)M_{A}^{i'}(\bthe_0)/\xi_{A}(i,\bphi) & \text{otherwise}.
\end{cases}
\end{align*}
Therefore the state transition probability matrix is
\[
\bm{\tgam}(A) = \bm{D}'(A)^{-1}\bm{\gamma}(A) \bm{D}''(A).
\]
For $A=\om$, the proof for the initial state probability vector is similar. The expression follows because the overall marginal likelihood is $\xi_{\om}(1,\bphi)$.
\end{proof}

\subsection*{S2.~Prior specification of the DPM of normals}
We use the {\tt DPpackage} function {\tt DPdensity} to carry out density estimation using the Dirichlet process mixture (DPM) of normals. The model formulation and the hyperparameter values are as follows, which follows an example in the user manual for {\tt DPpackage}
\begin{align*}
x_i\,|\,\mu_i,\Sigma_i &\sim {\rm N}(\mu_i,\Sigma_i) \quad \text{for $i=1,2,\ldots,n$}\\
(\mu_i,\Sigma_i)\,|\,H &\sim H\\
H\,|\,\alpha,H_0 &\sim {\rm DP}(\alpha H_0)\\
H_0 &= {\rm N}(\mu|m_1,\Sigma/k_0)\times {\rm IW}(\Sigma|\nu_1,\psi_1) \\
\alpha\,|\,a_0,b_0 &\sim {\rm Gamma}(a_0,b_0)\\
m_1\,|\,m_2,s_2 &\sim {\rm N}(m_2,s_2)\\
k_0\,|\,\tau_1,\tau_2 &\sim {\rm Gamma}(\tau_1/2,\tau_2/2)\\
\psi_1\,|\,\nu_2,\psi_2 &\sim {\rm IW}(\nu_2,\psi_2)
\end{align*}
where $a_0=2$, $b_0=1$, $m_2=0$, $s_2=10^5$, $\psi_2={\rm diagonal}(0.5,1)$, $\nu_1=4$, $\nu_2=4$, $\tau_1=1$, and $\tau_2=100$. We draw 5,000 posterior samples using 1,000 burn-in iterations and a 10-iteration thinning window.
\end{document}